\documentclass{lmcs}
\pdfoutput=1

\usepackage{lastpage}
\lmcsdoi{22}{2}{8}
\lmcsheading{}{\pageref{LastPage}}{}{}%
{Jan.~10,~2023}{Apr.~21,~2026}{}

\usepackage[utf8]{inputenc}
\usepackage{microtype}
\usepackage[hidelinks,breaklinks]{hyperref}
\usepackage{url}

\title[Tighter bounds for query answering with guarded TGDs]{Tighter bounds for query answering\texorpdfstring{\\}{} with guarded TGDs}
\author{Antoine Amarilli\lmcsorcid{0000-0002-7977-4441}}[a]
\address{Univ.\ Lille, Inria Lille, CNRS, Centrale Lille, UMR
9189 CRIStAL, F-59000 Lille, FR;\newline LTCI, T\'el\'ecom Paris, Institut Polytechnique de Paris, FR}
\email{antoine.a.amarilli@inria.fr}
\author{Michael Benedikt\lmcsorcid{0000-0003-2964-0880}}[b]
\address{Department of Computer Science, Oxford University, Parks Rd, Oxford OX1 3QD, UK}
\email{michael.benedikt@cs.ox.ac.uk}

\usepackage{multirow}
\usepackage{amsmath}
\usepackage{amsthm}
\usepackage{thm-restate}
\usepackage{xspace}
\usepackage{amsfonts}
\usepackage{tabularx}
\usepackage{booktabs}
\usepackage{xcolor}

\newcommand{\body}{\beta}
\newcommand{\head}{\eta}

\newcommand{\owqa}{\kw{OWQA}}

\newcommand{\datalogpm}{\kw{Datalog}^\pm}
\newcommand{\inst}{I}

\newcommand{\lift}{\kw{Lin}}
\newcommand{\linearize}{\kw{Linearize}}
\newcommand{\delinearize}{\kw{Delinearize}}

\newcommand{\ptime}{\kw{PTIME}}
\newcommand{\exptime}{\kw{EXPTIME}}
\newcommand{\expspace}{\kw{EXPSPACE}}

\newcommand{\adom}{\kw{Adom}}
\newcommand{\dom}{\adom}

\newcommand{\twoexp}{\kw{2EXPTIME}}

\newcommand{\kw}[1]{{\mathsf{#1}}\xspace}

\newcommand{\accessible}{\kw{accessible}}

\newcommand{\pspace}{\kw{PSPACE}}

\newcommand{\myeat}[1]{}

\newcommand{\card}[1]{\left|#1\right|}

\newcommand{\NN}{\mathbb{N}}

\newcommand{\np}{\kw{NP}}

\newcommand{\wclo}[1]{\widehat{#1}}

\newcommand{\sign}{{\mathcal{S}}}
\newcommand{\sidesign}{\sign'}

\newcommand{\dep}{\delta}
\newcommand{\gdep}{\gamma}
\newcommand{\trig}{\tau} %

\newcommand{\triv}{\mathrm{triv}}

\thanks{We are very grateful to the anonymous referees of LMCS for several rounds of detailed feedback and corrections.}	

\usepackage{colonequals}

\newcommand{\myparagraph}[1]{\paragraph*{#1.}}

\newcommand{\myproof}{Proof }

\hypersetup{
    colorlinks,
    linkcolor={red!50!black},
    citecolor={blue!50!black},
    urlcolor={blue!30!black}
}

\begin{document}
\begin{abstract}
We consider the complexity of the open-world query answering problem, where we
  wish to determine certain answers to conjunctive queries over
incomplete datasets specified by an initial set of facts
and a set of guarded TGDs.
This problem has been well-studied in the literature and is decidable but with a
  high complexity, namely, it is $\twoexp$ complete. Further, the complexity
shrinks by one exponential when the arity is fixed.

  We show in this paper how we can obtain better complexity bounds
when considering separately the arity of the guard atom and that of the additional
atoms, called the \emph{side signature}. Our results make use of the technique
of linearizing guarded TGDs, introduced in~\cite{gmp}.
Specifically, we present a variant of the linearization process, making use
of a restricted version of the chase that we recently introduced~\cite{resultlimitedj}. Our results imply that open-world query answering
  with guarded TGDs can
  be solved in $\exptime$ with arbitrary-arity guard relations if we simply
  bound the arity of the side signature; and that the complexity drops to $\np$
  if we fix the side signature and bound the width of the
  dependencies.
\end{abstract}

\maketitle

\section{Introduction} \label{sec:intro}

The \emph{open-world query answering} problem (OWQA) is concerned with evaluating a
query $Q$ on an incomplete dataset $\inst$. We specify $\inst$ as a set of
possible completions of an initial set of facts $I_0$, constrained by a set of
integrity constraints $\Sigma$: formally, we evaluate~$Q$ over all completions
of~$I_0$ satisfying~$\Sigma$. Typically $Q$ is a conjunctive query (CQ), equivalent to a basic
SQL query; the rules $\Sigma$ limit the possible completions to consider, and
are typically expressed in restricted logical formalisms. Research has therefore
sought to explore the tradeoff between the expressiveness of the constraint
language and the complexity of the open-world query answering problem, which we simply
call \emph{query answering} for brevity.

In the last decades, 
one notable focus of research on query answering is the $\datalogpm$ family of
constraints: in this work we study the related
formalism of \emph{tuple-generating dependencies} (TGDs).
A TGD is a universally-quantified implication, where the left-hand side~--
the \emph{body}~-- is 
a conjunction of atoms, and the right-hand side of the implication~-- the \emph{head}~-- is a CQ.
Up to rewriting TGDs, we can ensure that their right-hand side is always a
single atom: these are called \emph{single-headed} TGDs and we focus on such
TGDs in this paper, though we refer back to the case of multi-headed TGDs at the end of Section~\ref{sec:results}.
A TGD is \emph{full} if its head contains no existential quantifiers. Such TGDs are sometimes called \emph{Datalog rules}.

TGDs are a common constraint language to express that some patterns in the data
must imply the existence of other patterns; it subsumes, e.g., \emph{inclusion
dependencies} (IDs) which are common in relational databases: inclusion dependencies are
TGDs where the body and head contain a single atom with no repeated variables.
The formalism of guarded TGDs (GTGDs) restricts TGDs by enforcing that there is one atom in the
body containing all the variables used in the body: such an atom is called a
\emph{guard atom}. This ensures that GTGDs are expressible in the guarded
fragment of first-order logic. Query answering is known to be decidable in
$\twoexp$ for this class, and better bounds are known in some special cases: for
example, when the arity of the signature is bounded, the complexity drops to
$\exptime$. Similarly, it is known that 
the complexity of query answering drops to
$\pspace$
for TGDs that are
\emph{linear}~\cite{johnsonklug,lukasiewicz2015classical}, i.e., whose body
consists of a single atom. Note that linear TGDs are a special case of GTGDs,
and IDs are themselves a special case of linear TGDs.
See Table~\ref{fig:summary} in Section~\ref{sec:results} for a more detailed comparison.

Our goal in this paper is to show finer bounds on query answering.
Our approach to do so is to distinguish between the relations allowed
for the guard atoms in TGDs and the signature used for the remaining atoms.
Specifically, 
we allow each GTGD to have an unrestricted guard atom, but restrict the other
atoms, e.g., by bounding their maximal arity, or fixing the \emph{side signature} from
which they are taken.
Roughly speaking,
our results show that the first limitation
suffices to bring the complexity of query answering
down to $\exptime$.
The second limitation
further lowers the complexity to~$\np$ if we also impose another restriction,
namely, bounding the so-called
\emph{width} of the GTGDs. The width is the number of \emph{exported variables}
of a GTGD, where a variable is
exported when it appears in both the head and the body.
Formally, the main complexity results that we show in this paper are the
following,
which apply to GTGDs that \emph{obey} a side signature, i.e., that only use the
relations of the side signature in rule bodies except that there may be one
guard atom on a non-side-signature relation (the formal definition will be
given as Definition~\ref{def:side}):

\begin{restatable}{result}{resexptime}
  \label{res:exptime}
  The $\owqa$ problem with GTGDs that obey a side signature is in $\exptime$ if we assume
  that the arity of the side signature is bounded by a constant.
\end{restatable}

\begin{restatable}{result}{resnp}
  \label{res:np}
  The $\owqa$ problem with GTGDs that obey a side signature is in $\np$ 
  if we assume that the 
  side signature is fixed and that the
  width of the GTGDs is bounded by a constant.
\end{restatable}

Note how the first result provides a unifying language to recapture the $\exptime$ membership of query answering
with guarded TGDs of bounded arity, and the lower complexity of linear TGDs
(though we do not achieve a $\pspace$ but $\exptime$ bound).
That result is also tight in the sense that there are GTGDs with a fixed side
signature of arity $1$
for which 
$\owqa$ is $\exptime$-complete \cite{bbbicdt}.

The finer-grained bounds presented in this paper were first announced
in the context of our earlier work on \emph{access
patterns}~\cite{resultlimitedpods}, where we want to evaluate a query by
accessing the underlying relations in a limited way.
Research on access patterns typically relies on encoding problems with GTGDs that capture
the semantics of access methods.
For each access method,  the resulting TGDs will have a body
consisting of a single high-arity guard atom, and multiple unary
atoms over a particular predicate $\accessible$, intuitively denoting that some
instance element can be retrieved using the access patterns.
Specifically, the conference version of our work \cite{resultlimitedpods}
claimed a number of bounds on query answering with access patterns, using
side-signature-based techniques. These
results were only sketched in the conference version (for which unpublished
appendixes claim a weaker form of our main result, Theorem~\ref{thm:idreduce}).
These results were  omitted from the later
journal version~\cite{resultlimitedj}, which only contain bounds shown in the
restricted setting of access patterns.
Similar techniques are also used in a later
conference paper \cite{privacyijcai} about data disclosure, by the second author
and other authors.

The goal of the present work is to fill this gap, and give a self-contained
presentation of our complexity results on query answering for GTGDs with a side
signature. This makes it possible to understand these TGD complexity results independently
from the application to access methods that motivated it.
 Indeed, we believe these results to be interesting in their own
right, and hope that they can be useful to show complexity bounds on query
answering in different contexts.

\myparagraph{Revisiting linearization}
Our key proof technique is that of \emph{linearization}, introduced
in~\cite{gmp}. The idea is to reduce to the setting of linear TGDs, recalling
that 
query answering for this class is known to be quite
well-behaved~\cite{johnsonklug}:
it is in $\pspace$ in general,
and in $\np$  when the width is bounded.
Restricting the side signature intuitively makes GTGDs ``translatable'' to
linear TGDs in an auxiliary signature, where a body atom in the auxiliary
signature represents a guard atom in the original signature along with a specific
choice of additional atoms from the restricted side signature.

While this basic idea is quite natural, the linearization process turns out to
be rather involved. In particular, before we linearize we need to pre-process the TGDs so they behave more like
linear TGDs. The key part of this is a \emph{saturation
procedure} which generates derived GTGDs from our original set. Once this saturation is applied, the process of reasoning
with the GTGDs becomes very similar to reasoning with linear TGDs.
The saturation process is similar to ones that have been investigated in Description Logics~\cite{motikthesis} and guarded TGDs~\cite{gsatvldbjournal} for determining whether a ground fact is entailed (also known as ``atomic query answering'').
Indeed, in the process of giving our results, we will explain how linearization is a natural way to extend saturation procedures to deal with conjunctive query answering.

The main technical challenge is to use the side signature carefully to  obtain a saturation whose size is not too large. In addition to the side-signature-aware construction of the saturation, we use a number of further ideas.  To argue that the new saturation is correct,
we must make use of a specialized version of the \emph{chase procedure} \cite{maier,onet}, which augments an instance to create new entailed facts.  The idea is that the saturation procedure iteratively combines multiple TGDs to get new full TGDs, and in doing so it mimics the way facts would propagate when performing a chase.  But our restricted saturation only mimics a restricted flow of facts, and we must argue that this suffices. We require that it is sufficient  to simulate a restricted kind of chase called the \emph{one-pass chase}. The one-pass chase is introduced as a tool for reasoning about completeness of saturation procedures in ~\cite{gsatvldb}. But here,  again, we will need to customize it to the side signature-aware context. See Section \ref{sec:onepass} for the formal definitions.

Our linearization algorithm will be fairly straightforward once we close under the saturation process. But to argue that it is correct we will need yet another variation of the chase procedure. The chase with linear TGDs has a special form where one does not perform propagation of facts up and down the tree. We need to argue that GTGDs also support a variation of the chase where propagation is very restricted, the \emph{shortcut chase}. See Section \ref{sec:shortcut} for the precise definition.

\myparagraph{Paper structure}
We first give preliminaries in Section~\ref{sec:prelims}, reviewing in
particular the notion of 
\emph{semi-width} from~\cite{resultlimitedj}.
We then define the notion of side signature in
Section~\ref{sec:results} and state our main result (Theorem~\ref{thm:idreduce}), which describes the
complexity of translating GTGDs to linear TGDs of bounded semi-width, depending on
bounds on the side signature and the width of the GTGDs. Together with the results on
semi-width, this result directly implies Results~\ref{res:exptime} and~\ref{res:np} above. 
The rest of the paper is devoted to giving the proof of Theorem~\ref{thm:idreduce}. We start by giving some normalizations of our GTGDs (Section~\ref{sec:simplify}) and chase proofs (Section~\ref{sec:onepass})
that will be useful.
We give a proof overview of the main result in Section~\ref{sec:proofoverview}, and
the proof itself spans Sections~\ref{sec:saturation}, \ref{sec:shortcut}, and~\ref{sec:linearize}.
We conclude in Section~\ref{sec:conc}.

\section{Preliminaries} \label{sec:prelims}

\myparagraph{Data and queries}
We consider a \emph{relational signature} $\sign$ that consists of
a set of \emph{relations} with an associated
\emph{arity} (a positive integer).
The \emph{arity} of~$\sign$ is the maximal arity of a relation in~$\sign$.
The set of \emph{positions} of a relation~$R$ of~$\sign$ is the set 
$\{R[1]\ldots R[n]\}$ where
$n$ is the arity of~$R$.
An \emph{instance} of~$R$ is a
(finite or infinite) set of~$n$-tuples of values from some infinite set of
values; we also call these values the \emph{domain elements}, or simply
\emph{elements}.
Note that the domain elements will also include \emph{nulls}, which are the
fresh values created in the facts added
when performing \emph{chase steps} with non-full tuple-generating dependencies
(this will be formally defined in Section~\ref{sec:onepass}).
An \emph{instance} $I$ of~$\sign$
consists of instances
for each relation of~$\sign$.
The \emph{active domain} $\adom(I)$ of~$I$
is the set of the
domain elements that occur in tuples of~$I$.

A \emph{$\sign$-atom}, or simply \emph{atom}, is an expression of the form
$R(x_1 \ldots x_n)$, where $R$ is a relation of $\sign$ and $n$ is the arity
of~$R$ in~$\sign$.
We also call $R(x_1 \ldots x_n)$ an \emph{$R$-atom}.
We will be concerned mostly with two kinds of atoms,
depending on the nature of the terms $x_1, \ldots, x_n$: those
where all the terms are variables, and those where
all the terms are domain elements. We call the latter a \emph{ground atom} or a
\emph{fact}. For a fact $F = R(\vec a)$, we write $\adom(F)$ for the set of
elements that occur in~$F$ (i.e., those of~$\vec a$), and we also call $F$ an
\emph{$R$-fact}.
We will equivalently see instances~$I$ as a set
of \emph{facts} $R(a_1 \ldots a_n)$ for each tuple 
$(a_1 \ldots a_n)$ in the instance of each relation~$R$.
A \emph{subinstance} $I'$ of~$I$ is then an instance that contains a subset of the
facts of~$I$: we equivalently say that $I$ is a \emph{superinstance} of $I'$.

A \emph{homomorphism} from a set of atoms $A_1, \ldots, A_n$ to a set of atoms~$B_1, \ldots, B_m$ is a substitution $\sigma$ of the variables occurring in the~$A_i$ 
such that, for each atom~$A_i$, writing
$\sigma(A_i)$ the result of substituting its elements according to~$\sigma$,
then the result is one of the~$B_1, \ldots, B_m$.

We will study \emph{Boolean conjunctive queries} (CQs), which are logical expressions of the form
$\exists x_1 \ldots x_k  ~ (A_1 \wedge \cdots \wedge A_m)$, where
the~$A_i$ are atoms over $x_1 \ldots x_k$.
Note that we only focus on \emph{Boolean CQs}, which have no free variables.
For this reason, throughout the paper, \emph{by default when we say
a CQ we mean a Boolean one}. Also note that we do not allow
constants in CQs, but these can be encoded, e.g., by expanding the signature
with additional unary predicates to distinguish the constants in instances.
We further discuss the impact of constants in Section~\ref{sec:results}.
A \emph{match} of a CQ $Q$ 
in an instance $I$ is a 
homomorphism from~$Q$ to~$I$, i.e., a mapping~$h$ from the variables 
of~$Q$ to~$\adom(I)$ which 
ensures that,
for every atom $R(t_1 \ldots t_n)$ in~$Q$, we have that $R(h(t_1) \ldots h(t_n))$ is a fact of~$I$.
We say that $Q$ \emph{holds} in~$I$ if there is a match of~$Q$ in~$I$.

\myparagraph{Integrity constraints}
We study \emph{integrity constraints} which are defined in (fragments of) first-order logic
(FO), disallowing constants.
In our definition of FO, we only consider
FO formulas and FO fragments where all quantified variables 
that appear in a formula must appear in some relation of the formula.
This ensures that the satisfaction of a constraint on an instance $I$ only depends
on the active domain, i.e., on the values occurring in facts of~$I$; in other
words, we follow the active-domain semantics.
For an FO formula $\rho$ and an instance $I$, we say that $I$ \emph{satisfies}
$\rho$, written $I \models \rho$, if $\rho$ holds on~$I$ with the usual
semantics of FO; we omit the corresponding definitions (see, e.g., \cite{libkin1995elements}). Otherwise,
we say that $I$ \emph{violates} $\rho$, written $I \not\models \rho$.
For $\Sigma$ a set of FO formulas, we write $I \models \Sigma$ to say that $I$
satisfies all formulas of~$\Sigma$.

We  focus on 
a specific fragment of FO,
called
\emph{tuple-generating dependencies} (TGDs),
which we now review.
A \emph{tuple-generating dependency} (TGD) is an FO sentence~$\gdep$ of the form:
$\forall \vec x ~ (\body(\vec x) \rightarrow \exists \vec y ~ A(\vec x, \vec
y))$
where $\body$ is a conjunction of atoms called the \emph{body} and where all
variables from $\vec x$ appear, and $A$ is an atom
called the \emph{head}.
We require that all variables of $\vec x$ occur in $\body$ and that all
variables of $\vec y$ occur in~$A(\vec x, \vec y)$, but $A(\vec x, \vec y)$ may
only use a subset of the variables of~$\vec x$ (or possibly none at all).
Note that we define here 
\emph{single-headed} TGDs, whose head consists of a single atom. \emph{We focus on
single-headed TGDs throughout this paper}: see for instance \cite{gottlob2020multi} for an
example of work investigating the impact of this choice, and see the end of
Section~\ref{sec:results} for a further discussion of the matter.

For $\gdep$ a TGD and $I$ an instance, using the previous notation, a substitution $\trig$ from $\body$ to~$I$ is called a \emph{trigger} of~$\gdep$ in~$I$. The trigger~$\trig$ is said to be \emph{active} if there is no mapping $\trig'$ that extends~$\trig$ and ensures that $\trig'(A(\vec x,\vec y)$ occurs in~$I$. The semantics of~$\gdep$ is that $I$ satisfies~$\gdep$ if and only if there are no active triggers of~$\gdep$ in~$I$.

For brevity, in the sequel, we will \emph{omit outermost universal quantifications in TGDs}.
The \emph{exported variables} of~$\gdep$ (also called \emph{frontier
variables}) are the variables of~$\vec x$ which
occur in the head.
A \emph{full TGD} is one with no existential quantifiers in the head.
A \emph{guarded TGD} (GTGD) is a TGD~$\gdep$ whose
body $\body$ contains some atom $A$ which contains all
variables occurring in~$\body$. We call $A$ a \emph{guard} of~$\body$, and
of~$\gdep$: note that it is not necessarily unique.
When decomposing $\body$ as $A(\vec x) \wedge \body'(\vec x)$ for one specific
choice of~$A$, we call $A$ the \emph{guard atom}.

We will also say that a fact or set of facts $S$ is \emph{guarded} by
another fact or set of facts~$S'$ if the facts of~$S$ only use values
occurring in~$S'$.

A \emph{linear TGD} is a TGD where the body consist of a single atom:
remember that we defined TGDs to be single-headed so we already know that the
head is also a single atom. Also notice that a linear TGD is always guarded. An
\emph{inclusion dependency} (ID) is a linear TGD where we further impose
that no variable is
repeated in the body and no variable is repeated in the head.

The \emph{width} of a TGD is the number
of exported variables. Remember that we do not allow constants in TGDs.
\myparagraph{Fact entailment, $\owqa$, and TGD entailment problems}
We refer throughout the paper to the standard notion of entailment in first-order logic:
\begin{defi}
We say that  a set of FO sentences $\lambda$ \emph{entails} an FO sentence
$\rho$, written $\lambda \models \rho$, if every  instance satisfying $\lambda$ also satisfies
$\rho$.
\end{defi}
In particular, a special case of entailment is \emph{entailment of a TGD $\tau$
by a set of TGDs~$\Sigma$}, which we will use in the technical proofs.
Another case, which is the main focus of this paper, is that of entailment
problems of the form:
\[
\bigwedge_{i \leq n} F_i \wedge \Sigma \models Q
\]
where $\Sigma$ is a set of TGDs,
each $F_i$ is a fact, and $Q$ is a
CQ.
This is  the problem of \emph{certain answers} or \emph{(open-world) query
answering}~\cite{fagindataex} ($\owqa$)
 under TGDs for CQs. 

\begin{defi}
 If $\inst_0$ is a finite instance with facts $F_1 \ldots F_n$,  we also
write $\inst_0, \Sigma \models Q$ to mean $\bigwedge_i F_i \wedge \Sigma \models Q$.
\end{defi}

The \emph{$\owqa$ problem} is the problem of deciding whether such
entailments hold. Formally, the input to $\owqa$ consists of a finite
instance~$\inst_0$, a set $\Sigma$ of TGDs, and a CQ~$Q$; the output is a
Boolean indicating whether $\inst_0, \Sigma \models Q$
or $\inst_0, \Sigma \not\models Q$.

One variant of $\owqa$ that we will study is \emph{fact entailment}. In
this problem, the input consists of a finite instance $\inst_0$, a set $\Sigma$
of TGDs, and a fact $F$ on the domain of~$\inst_0$. The output is a Boolean
indicating whether $\inst_0, \Sigma \models F$, i.e., the fact $F$ is contained in
every superinstance of~$\inst_0$ that satisfies~$\Sigma$.

Note that the $\owqa$ problem is closely connected to the problem of \emph{query
containment under constraints}, which has been independently studied in the
literature, and also in connection with query containment under access patterns; see, e.g., \cite{dln}.

In this paper, we study the complexity of the $\owqa$ problem in \emph{combined complexity}, i.e.,
as a function of~$|\inst_0|$, $|\Sigma|$, and $|Q|$, where the size $|\Sigma|$ is
taken following, e.g., a string representation, and likewise for~$|\inst_0|$
and~$|Q|$.

We now discuss two ways in which
$\owqa$ problems can be equivalent, while possibly changing the underlying
signature. One first notion is
\emph{entailment-equivalence}, which is defined on
sets of TGDs:

\begin{defi}
For $\sign$ a signature,
we say that two finite sets of TGDs $\Sigma$ and $\Sigma'$ are
  \emph{\mbox{$\sign$-entailment-equivalent}}
  if they are interchangeable for $\owqa$
on~$\sign$, namely:
for any instance $\inst_0$ over~$\sign$ and CQ $Q$ over~$\sign$, we have $\inst_0, \Sigma \models Q$ iff $\inst_0,
\Sigma' \models Q$.
\end{defi}

Note that $\Sigma$ and $\Sigma'$ may be on a larger signature than~$\sign$. 

\begin{rem}
Note
that entailment-equivalence is weaker than logical equivalence. Indeed, if $\Sigma$
and $\Sigma'$ are logically equivalent in the sense that $\Sigma \models
\Sigma'$ and vice-versa, then they are $\sign$-entailment-equivalent over any
signature $\sign$. However, the converse is not true: for $\phi$ the GTGD $R(x) \rightarrow S(x)$, letting
$\Sigma = \{\phi\}$ and $\Sigma' = \emptyset$, and $\sign$ be a
signature containing $R$ but
not~$S$, then $\Sigma$ and $\Sigma'$ are not logically equivalent but they are
$\sign$-entailment-equivalent.
\end{rem}

A second notion is \emph{emulation}, which is defined on pairs of a set of TGDs
and of an instance:

\begin{defi}
  For $\sign$ a signature which is a subsignature of~$\sign'$,
we say that
a set $\Sigma'$ of constraints and an instance~$\inst_0'$ on
signature $\sign'$ \emph{emulates} another set of
constraints~$\Sigma$ and instance~$\inst_0$ on signature~$\sign$ if they are equivalent for $\owqa$, i.e.,
entail the same CQs on~$\sign$. Formally, $\inst_0'$ and $\Sigma'$ \emph{emulate} $\inst_0$ and
$\Sigma$ if, for
any CQ~$Q$ posed over the signature~$\sign$, 
we have $\inst_0, \Sigma \models Q$ iff $\inst_0', \Sigma' \models
Q$.
\end{defi}

\paragraph*{Semi-width.}
Our tractability results in this paper will be shown by reducing to the case of
linear TGDs, which we call \emph{linearization}. In particular, we will rely
on $\np$ upper bounds which depend on the \emph{width} of the dependencies.
The following result was shown by Johnson and Klug~\cite{johnsonklug} in the
case of inclusion dependencies:

\begin{propC}[\cite{johnsonklug}]
For any fixed $w \geq 0$,
there is an $\np$ algorithm for $\owqa$ 
under inclusion dependencies of
width at most~$w$.
\end{propC}

This notion was slightly generalized in~\cite{resultlimitedj} to the notion of
\emph{semi-width}, which extends width by the addition of acyclic TGDs. We now
review the definition.
The \emph{basic position graph} of a set of TGDs
$\Sigma$ is the directed graph whose nodes are the positions of relations
in~$\Sigma$,  with an edge from~$R[i]$ to~$S[j]$ if and only
if there is a rule $\dep \in \Sigma$ with exported
variable $x$ occurring at position~$i$ of an $R$-atom in the body
of~$\dep$ and at position~$j$ of an $S$-atom in the head of~$\dep$.
Note that our basic position graph is different from the notion of position
graph used to ensure decidability of termination  for general TGDs \cite{onet}.  We apply
the basic position graph only to complexity considerations concerning linear TGDs.
In particular, existentially quantified variables do not contribute to the basic
position graph.

We say that a collection of TGDs $\Sigma$ has \emph{semi-width} bounded
by $w$ if $\Sigma$ can be partitioned into~$\Sigma_1 \cup \Sigma_2$
where $\Sigma_1$ has width bounded by $w$ and where
the basic position graph of~$\Sigma_2$ is acyclic.

We can then show that $\owqa$ is in $\np$ for linear TGDs of bounded semi-width.
The result was shown in~\cite[Appendix~C]{resultlimitedj} for inclusion
dependencies \cite[Proposition~6.5]{resultlimitedj}, and we must slightly modify the proof to apply more generally to
linear TGDs:

\begin{restatable}{prprestate}{semiwidthclassic}
  \label{prop:semiwidthclassic-general} 
  For fixed  $w$,
  there is an $\np$ algorithm for OWQA
  under linear TGDs
  of semi-width at most~$w$.
\end{restatable}

To avoid distracting from the main results of the paper, we give a
self-contained proof of this result in Appendix~\ref{apx:semiwidthproof}, most of which is identical to~\cite[Appendix~C]{resultlimitedj}.

\section{Result statements}
\label{sec:results}

Having defined the preliminaries, we are now ready to formally state the results
of this paper. We first position our work relative to existing results. Then, we
introduce the
notion of \emph{side signature} 
through which our results are phrased. We
then
give the formal statement of the results. After this we state our main technical
linearization result, and explain how our main results follow from
that result. We finally discuss the issues of multi-headed GTGDs and
constants, and explain the relationship with earlier techniques.

\myparagraph{Earlier results}
Our focus is on complexity bounds for the $\owqa$ problem with guarded TGDs.
The following bounds on this problem  represent the prior state of the art \cite{tamingjournal}:

\begin{thmC}[\cite{tamingjournal,johnsonklug}]
  Given a set of guarded TGDs $\Sigma$, an instance~$I$, and a query~$Q$, the
  $\owqa$ problem for $\Sigma$, $I$, and $Q$ is $\twoexp$-complete.

  Further, if we fix the arity of the signature~$\sign$, then the problem is
  $\exptime$-complete. Last, if we fix the guarded TGDs $\Sigma$, the problem is
  $\np$-complete.

  If the TGDs are linear, the problem is $\pspace$-complete, and $\np$ when the
  width is bounded. This was proven in \cite{johnsonklug} only for the case of Inclusion Dependencies, but the proof extends straightforwardly to linear TGDs.
\end{thmC}

Our goal in this paper is to refine the $\exptime$
and $\np$ bounds by showing that they hold for more general constraint
languages, defined via the notion of \emph{side signature}.

The current state of the art and contribution are summarized in
Table~\ref{fig:summary}.

\begin{table*}
\centering
{
\centering
\begin{tabular}{ lr }
\toprule
{\bfseries GTGD restriction} & {\bfseries $\owqa$ complexity} \\
\midrule 
General & $\twoexp$ \cite{tamingjournal} \\
Fixed arity & $\exptime$ \cite{tamingjournal} \\
{\bf Fixed side signature arity} & {\bf $\exptime$} (Result~\ref{res:exptime}) \\
Linear & $\pspace$ \cite{johnsonklug} \\
  Fixed GTGDs  & $\np$ \cite{tamingjournal} \\
   Fixed-width IDs & $\np$ \cite{johnsonklug} \\
{\bf Fixed side signature and width}  & {\bf $\np$} (Result~\ref{res:np}) \\
\bottomrule 
\end{tabular}
\caption{Complexity results, with new results in bold} \label{fig:summary}
    }
    \end{table*}

\myparagraph{Side signature}
We now introduce the notion of \emph{side signature} that is fundamental to the statement of our results. The side signature intuitively consists of relations that can be
used together with a guard in rule bodies. More precisely, we consider GTGDs where, for some
choice of guard atom, all other atoms must use relations in the side signature:
\begin{defi}
  \label{def:side}
  Let $\gdep$ be a 
  GTGD on signature $\sign$.
  Given a subsignature~$\sidesign \subseteq \sign$, we say that~$\gdep$ 
  \emph{obeys side signature $\sidesign$}
  if there is a choice of guard atom in the body
  of~$\gdep$ such that all other body atoms are relations of~$\sidesign$. A set
  of GTGDs $\Sigma$ \emph{obeys side signature $\sidesign$} if all GTGDs of
  $\Sigma$ do.
\end{defi}

\begin{defi} \label{defi:principal}
We refer to the relations in $\sign \setminus \sidesign$ as \emph{principal
  relations}, as opposed to \emph{side relations}. We similarly refer to
  \emph{principal atoms} and \emph{principal facts}, versus \emph{side atoms}
  and \emph{side facts}, depending on whether the relation involved is
  a side relation or a principal relation.
  We also say that a TGD is a \emph{principal TGD}
  if its head atom is a principal atom, and a \emph{side TGD} otherwise.
\end{defi}
In particular, a linear TGD always obeys any choice of side signature.
Note that we can easily test whether GTGDs $\Sigma$ obey side signature $\sidesign$:
consider the body of each GTGD and check that it has at most one atom which is
not in $\sidesign$ and that this atom guards the body.
The choice of side
signature is also not canonical, e.g., taking $\sidesign := \sign$ always
satisfies the definition, but this trivial choice will not be useful to achieve good complexity
bounds using our results.

We give an example of the notion of side signature:

\begin{exa}
  \label{exa:side}
  Consider the following set of GTGDs over relations $\{R, S, T, U\}$.
  \begin{align*}
    R(x, y, x, z), T(x), T(z), U(x, z) & \rightarrow \exists w ~ S(y, w) \\
    U(x, y), U(x, x) &\rightarrow U(y, y)\\
    S(x, y), U(x, y) &\rightarrow S(y, x)
  \end{align*}
  They obey the side
  signature $\{T, U\}$.
  Note that obeying a side signature does not impose constraints on head atoms,
  which can be principal relations or side relations.
  Further, when a GTGD has a body that only uses relations from a given side
  signature, then it obeys that side signature.
\end{exa}

Our results will apply to sets of GTGDs $\Sigma$ that obey a
side signature $\sidesign$ and where other conditions are respected,
in particular the arity of $\sidesign$ will have to be bounded.

\paragraph*{Main results.}
We can now restate our two main results from the introduction. We show the
following, which gives a sufficient condition for $\owqa$ with GTGDs to be
in $\exptime$:

\resexptime*

In other words, we show that, for the $\owqa$ problem to be in $\exptime$, we can
allow arbitrary-arity guarded TGDs, and we just need to bound the arity of the
side signature. In particular, our result recaptures and extends the
lower complexity of $\owqa$ for linear TGDs~\cite{johnsonklug},
except that it shows an $\exptime$ bound rather than a $\pspace$ bound.

We further show that,
once the side signature is fixed, then
to achieve $\np$ complexity, we do not need to fix the
dependencies or even the arity~-- it suffices to fix the side signature
(including its arity) and the \emph{width} of the dependencies:

\resnp*

Our result extends the $\np$ upper bound on $\owqa$ for
bounded-width IDs shown by Johnson and Klug~\cite{johnsonklug},  as well as the $\np$ bound with fixed GTGDs \cite{tamingjournal}.

\begin{exa} \label{exa:access}
  We give an example from the setting of \emph{access patterns} which was
  mentioned in the introduction, explaining how Result~\ref{res:np} generalizes results proven in~\cite{resultlimitedj}.

Fix a number $m$,
consider a signature that includes a distinguished unary relation $\accessible$,
  and consider a set of TGDs of one of the two forms:
\begin{align*}
R(x_1 \ldots x_m, \vec y) \rightarrow \exists \vec z ~ H(x_{m_1} \ldots x_{m_k}, \vec z) \\
R(x_1 \ldots x_n) \wedge \bigwedge_{i \in S} \accessible(x_i) \rightarrow \accessible(x_j)
\end{align*}
In the first form of TGDs, $m_1 \ldots m_k$ are numbers bounded by $m$. These are \emph{linear TGDs} of \emph{width at most $m$}. In isolation, query answering is known to be $\np$ for such TGDs by a variation of \cite{johnsonklug}.

  In  TGDs of the second form,  $j$ is a number in $\{1 \ldots n\}$, and $S$ is an arbitrary set of such numbers. 
They are  \emph{full TGDs}, but they are not linear. These are referred to as \emph{accessibility axioms}
in \cite{resultlimitedpods,resultlimitedj}. 

Result~\ref{res:np} implies that query answering
  for the class of TGDs of this form is in $\np$. Intuitively, this is because the side signature
  is fixed and the width of TGDs is bounded. 
  This result is claimed (in the restricted case of inclusion dependencies,
  but in the broader context of result-bounded interfaces) in
  \cite[Theorem 6.4]{resultlimitedj}.
  One motivation for side signatures is to generalize this
  result but replacing $\accessible(x)$ by an 
  arbitrary fixed side signature, giving Result~\ref{res:np}.
\end{exa}

\myparagraph{Linearization result}
All of our complexity bounds are shown by reducing $\owqa$ with GTGDs to
$\owqa$ with linear TGDs.
We do this using our main technical result:

\begin{thm}
  \label{thm:idreduce}
  Let $a' \in \NN$ be a fixed bound on the side signature arity.
  There are polynomials $P_1$ and $P_2$ depending only on~$a'$ and 
  an algorithm with the following input:
  \begin{itemize}
  \item A signature $\sign$, where we let $a$ be the arity of~$\sign$;
  \item A subsignature $\sidesign \subseteq \sign$ of arity~$a'$, where we let~$n'$
    be the number of relations of $\sidesign$;
  \item An instance $I_0$ of~$\sign$;
  \item A %
  set $\Sigma$ of 
    GTGDs obeying side signature
    $\sidesign$, where we let $w$ be the maximal width of a GTGD of~$\Sigma$ and
      let $w' := \max(a', w)$.
  \end{itemize}
  The algorithm computes in time
  $P_1(\card{\Sigma} \times a \times \card{I_0})^{P_2(w', n')}$
  a set $\Sigma'$ of linear TGDs 
  of semi-width $\leq w'$
  and arity~$\leq a$,
  and an instance $I_0^\lift$, such that
  $\Sigma'$ and $I_0^\lift$ emulate~$\Sigma$ and~$I_0$ on signature~$\sign$.
\end{thm}

Theorem~\ref{thm:idreduce} will be proven in the next sections.
Before this, we show how Theorem~\ref{thm:idreduce} allows us to prove
Result~\ref{res:exptime} and Result~\ref{res:np}.

\myparagraph{Proving Result~\ref{res:exptime} from Theorem~\ref{thm:idreduce}}
Recall the statement of Result~\ref{res:exptime}:

\resexptime*

\begin{proof}[\myproof of Result~\ref{res:exptime}]
  We apply the reduction of Theorem~\ref{thm:idreduce}, which computes in
  $\exptime$ a finite set of linear TGDs $\Sigma'$ and rewritten instance $I_0^\lift$
  that emulate the original GTGDs $\Sigma$ and instance~$I_0$. 
We can solve the $\owqa$ problem for linear TGDs via an algorithm that takes
$\exptime$ 
in the maximum arity, with the running time of the algorithm being bounded
by a polynomial in the input instance
and the TGDs.
  This can be done
either by the chase-based argument of \cite{johnsonklug} or via an algorithm that iteratively performs query-rewriting to generate a derived UCQ and then evaluate it,
where each CQ has polynomial size
in the input CQ and maximal arity of the linear TGDs. Note that the number of such CQs
is polynomial in the number of linear TGDs and exponential in the arity
\cite{calirewriting}. When we apply this to our exponential set of
linear TGDs, we get the desired $\exptime$ bound. 
\end{proof}

Recall now the statement of Result~\ref{res:np}:

\resnp*

\begin{proof}[\myproof of Result~\ref{res:np}]
  We apply the reduction of Theorem~\ref{thm:idreduce} to obtain a finite set of linear TGDs $\Sigma'$ and rewritten instance $I_0^\lift$
  that emulate the original GTGDs $\Sigma$ and instance~$I_0$, and under our
  assumptions the running time bound of the theorem is in
  $\ptime$ because $w'$ and $n'$ are constant.
  Then, we conclude directly by Proposition~\ref{prop:semiwidthclassic-general}.
\end{proof}

\myparagraph{Constants, multi-headed dependencies, and IDs}
Note that we have stated our results with dependencies that are 
\emph{single-headed}, and which do not feature \emph{constants}. We have also
assumed that the query is Boolean ---
though non-Boolean queries can be encoded
instead as a Boolean query where the output variables are instantiated with
constants.

For the $\exptime$ upper bound (Result~\ref{res:exptime}), we do not expect that
these restrictions make a difference. Intuitively, constants can be emulated by
adding unary relations to the side signature (so without increasing its arity),
and multi-headed rules can be rewritten to be single-headed. See
Appendix~\ref{apx:constantmulti} for details about this process.
Our precise claim is that the result can be extended to multi-headed GTGDs with constants
that may be present in the query and in rule bodies -- we leave open the case of
GTGDs featuring constants in rule heads.

For the $\np$ upper bound (Result~\ref{res:np}), rewriting constants and multi-headed dependencies may increase the width and the side signature,
and we do not
know if the bound still holds if we allow constants or allow multi-headed
dependencies.

\myparagraph{Inapplicability of hardness results}
We also give a short explanation of why our results do not contradict the known
hardness results of~\cite{tamingjournal} on $\owqa$ with GTGDs.

For Result~\ref{res:exptime},
it is shown in~\cite[Theorem~6.2]{tamingjournal} that 
the $\owqa$ problem on a fixed instance for an atomic query under GTGDs is $\twoexp$-hard when the arity is
unbounded, even when the number of relations in the signature is bounded. The
proof works by devising a set of GTGDs that simulates an $\expspace$ alternating
Turing machine, by coding the state of the Turing machine as facts: %
specifically, a fact $\mathit{zero}(\mathbf{V}, X)$ codes that there
is a zero in the cell indexed by the binary vector $\mathbf{V}$ in
configuration~$X$. The arity of such relations is unbounded, so they cannot be
part of the side signature~$\sidesign$. However, in the simulation of the Turing
machine, the GTGDs in the proof use another relation as guard (the~$g$
relation), and the bodies contain other high-arity relations. Thus, 
there is no
choice of side signature of bounded arity which is obeyed by 
the set of GTGDs defined in the hardness proof of~\cite{tamingjournal}.

For Result~\ref{res:np}, the proof in~\cite[Theorem~6.2]{tamingjournal}
explicitly writes the state of the~$i^{th}$ tape cell of a configuration~$X$ as,
e.g., $\mathit{zero_i}(X)$. These relations occur in rule bodies where they are
not guards, but as Result~\ref{res:np} assumes that the side signature is
fixed, they cannot be part of the side signature. A variant of the construction
of the proof (to show $\exptime$-hardness on an unbounded signature arity) would be
to code configurations as tuples of elements $X_1 \ldots X_n$ and write, e.g.,
$\mathit{zero}(X_i)$. However, the constant-width bound on GTGDs would then mean
that the proof construction can only look at a constant number of cells when
creating one configuration from the previous one.

\paragraph*{Roadmap.}
The rest of this paper is devoted to proving
Theorem~\ref{thm:idreduce}, from which we already explained how to derive our
main results (Result~\ref{res:exptime} and~\ref{res:np}). We first present two
tools:
a normalization of GTGDs that obey a side signature (in
Section~\ref{sec:simplify}) and a chase process (in Section~\ref{sec:onepass})
which generalizes the ``one-pass chase'' from earlier
work~\cite{gsatvldbjournal} to be aware of the side signature. With these tools in place, we give a proof overview in Section~\ref{sec:proofoverview}. The proof proper is spread out over Sections~\ref{sec:saturation},  \ref{sec:shortcut}, and~\ref{sec:linearize}.  We conclude in
Section~\ref{sec:conc}.

\section{Simplifying GTGDs obeying a side-signature} \label{sec:simplify}

In this section, we show a way to simplify sets of GTGDs that obey a side
signature, to enforce three restrictions over them which will simplify subsequent
proofs.

The first restriction that we will want to enforce is \emph{homomorphism-closure},
which will intuitively save us from having to think about how two different
exported variables may be mapped to the same element when firing a chase step: 

\begin{defi}
Given a GTGD $\sigma$ and a function $h$ from the exported variables of $\sigma$
  to the exported variables of~$\sigma$,
  we call $h(\sigma)$ the GTGD
produced by applying $h$ to every exported variable.
  A set $\Sigma$ of GTGDs is said to be \emph{homomorphism-closed} if for any GTGD $\dep$ in~$\Sigma$, and mapping $h$ as above,
  $h(\dep)$ is in $\Sigma$.
\end{defi}
Note that $h(\sigma)$ is a logical consequence of $\sigma$ for any function $h$ on the exported variables.

The second restriction that we will want to enforce is that GTGDs have precisely
one \emph{principal guard}:

\begin{defi}
  Let $\gdep$ be a GTGD obeying side signature $\sidesign$. A \emph{principal
  guard} of~$\gdep$ is a guard atom of the body of $\gdep$ whose relation is a principal
  relation, recalling that this means a relation not in $\sidesign$.
\end{defi}

Note that GTGDs obeying the side signature always have a principal guard if
their body contains a principal atom, because this atom can then be used as a
principal guard and the GTGD body cannot contain any other principal atom by
definition of obeying a side signature. The issue is that we may have GTGDs
whose body does not contain principal atoms at all. We will enforce the second
restriction by rewriting the GTGDs to enforce
that all GTGD bodies contain exactly one principal atom, and hence a principal guard.

The third restriction applies to non-full GTGDs: we want to make sure that every
non-full GTGD is a principal GTGD, i.e., its head atom is a principal atom
(Definition~\ref{defi:principal}).

We can now define the normal form that we want to ensure:

\begin{defi}
  Let $\sidesign$ be a subsignature. We say that a set of GTGDs $\Sigma$
  \emph{strongly obeys} side signature $\sidesign$ if it obeys $\sidesign$ and
  further:
  \begin{itemize}
    \item $\Sigma$ is homomorphism-closed
    \item Every GTGD of~$\Sigma$ has exactly one principal guard.
    \item For every non-full GTGD of~$\Sigma$, the head atom is a principal
      atom.
  \end{itemize}
\end{defi}

Our goal in this section is to show the following result:

\begin{prop}
  \label{prp:strongobey}
  Let $\sign$ be the signature, and let $\sidesign$ be a subsignature of~$\sign$. Let
  $\Sigma$ be a set of 
  GTGDs
  that obeys $\sidesign$, and
  let $I_0$ be an instance over~$\sign$.
  Let $w$ be the width of~$\Sigma$, and $a'$ the arity of~$\sidesign$.

  We can
  compute in time polynomial in $|I_0|$, $|\Sigma|$, and $2^w$ a signature
  $\sign'' \supseteq
  \sign$, a set $\Sigma'$ of 
  GTGDs over $\sign''$, and an
  instance $I_0'$ over~$\sign''$, such that:
  \begin{itemize}
    \item $|\sign''|$ is polynomial in~$|\sign|$ and~$|\Sigma|$;
    \item The width of~$\Sigma'$ is at most $\max(a', w)$;
    \item $\Sigma'$ strongly obeys $\sidesign$;
    \item $I_0', \Sigma'$ emulates $I_0, \Sigma$ on signature~$\sign$.
  \end{itemize}
\end{prop}

Thanks to this result, towards showing Theorem~\ref{thm:idreduce}, we can first
apply the result and assume that the input 
GTGDs $\Sigma$
strongly obey side signature~$\sidesign$.

We prove Proposition~\ref{prp:strongobey} in the rest of this section, in
successive steps.

\subsection{Principal guards and principal head atoms}
The first step to prove
Proposition~\ref{prp:strongobey}
is to enforce the condition on principal guards, which will also incidentally
imply the condition on non-full GTGDs. This is the only step that will
modify the input instance $I_0$.

Some GTGDs of~$\Sigma$ have a
body already featuring an atom on a principal relation, in which case they already have a
principal guard, which is unique because $\Sigma$ obeys $\sidesign$. However,
other GTGDs of~$\Sigma$ do not have any atom on a principal relation in their
body. To ensure that such GTGDs have a principal atom, we will add new principal
relations that can be used as principal guards.

Let us expand the signature $\sign$ to~$\sign''$ by creating, for each side
relation $R$ in $\sidesign$, a new principal relation $R'$ in $\sign''\setminus
\sidesign$ with same arity as~$R$.
Let us modify the GTGDs of~$\Sigma$ as follows: for each side relation $R$,
in every GTGD of~$\Sigma$ with an $R$-atom in the head, replace it with an
$R'$-atom on the same variables.
Further, for each side relation $R$, let us add to~$\Sigma'$ the full GTGD $R'(\vec x) \rightarrow R(\vec x)$. 
Last, in every GTGD $\gdep$ of~$\Sigma$ which does not have a principal guard, pick a
guard atom $A$ on some side relation $R$, and replace $\gdep$ in~$\Sigma'$ by a
rule $\gdep'$ obtained from~$\gdep$ by adding to the body of $\gdep$ an atom
$A'$ on the same elements as~$A$ with the principal relation~$R'$, which will serve as principal guard.

Lastly, to rewrite the instance $I_0$ to~$I_0'$, we do the following: for each
fact $R(\vec a)$ on a side relation $R$, we add the fact $R'(\vec a)$.

We first claim that the transformation is correct, namely:

\begin{clm}
  $I_0',\Sigma'$ emulates $I_0,\Sigma$ on signature~$\sign$.
\end{clm}

\begin{proof}
Assume first that a query $Q$ on~$\sign$ is not entailed by $I_0,\Sigma$, i.e., there is a counterexample superinstance $I$ of~$I_0$ which satisfies~$\Sigma$ and does not satisfy~$Q$. Then we build $I'$ by adding the fact $R'(\vec a)$ for every side fact $R(\vec a)$ of~$I$. The query~$Q$ is still not satisfied by~$I'$ because the restriction of $I$ and~$I'$ to relations of~$\sign$ is identical. Further, $I'$ is a superinstance of~$I_0'$, and it is easy to see that $I'$ satisfies~$\Sigma'$ because $I$ satisfies~$\Sigma$.

For the converse direction, let $Q$ be a query which is not entailed by~$I_0',\Sigma'$, let $I'$ be a counterexample model, and build $I$ from~$I'$ by removing all facts of $\sign''\setminus\sign$. Then $I$ is a superinstance of~$I_0$ that does not satisfy~$Q$. To see why $I$ satisfies~$\Sigma$, let $\gdep$ be a GTGD of~$\Sigma$ and let $\trig$ be a trigger of~$\gdep$ in~$I$. There is a corresponding GTGD $\gdep'$ in~$\Sigma'$ obtained by possibly adding one principal guard atom, and possibly changing the head. We claim that $\trig$ is also a trigger of~$\gdep'$ in~$\Sigma'$. Indeed, in $I'$, for every side fact $R(\vec a)$, the fact $R'(\vec a)$ also exists. This is by construction of $I_0'$ for the facts of~$I_0'$, and for the other facts of~$I'$ it is because $\Sigma'$ ensures that facts $R(\vec a)$ can only be derived from the GTGD $R'(\vec x) \rightarrow R(\vec x)$. Thus, $\trig$ is also a trigger of~$\gdep'$, i.e., the possibly extra atom in the body of~$\gdep'$ is also mapped by~$\trig$. Thus, as $I'$ satisfies $\Sigma'$, and together with rules of the form $R'(\vec x) \rightarrow R(\vec x)$, we know that the head of $\gdep$ also exists in~$I$.
\end{proof}

Now, the resulting $\Sigma'$ 
still obeys the side signature. Further,
the new signature $\sign''$ is such that $|\sign''| \leq 2 |\sign|$, and the
process is polynomial in~$|\Sigma|$ and in~$|I_0|$. The width of GTGDs
of~$\Sigma'$ is at most the width of the GTGDs of~$\Sigma$, except we added full
GTGDs (from $\sign'' \setminus \sign$ to $\sidesign$) whose width is~$a'$. Thus,
the result of this transformation satisfies the conditions of 
Proposition~\ref{prp:strongobey}, and
ensures that each GTGD now has a principal guard.

Note that, from the way we changed the dependencies, every non-full GTGD is a
principal GTGD, i.e., has a head atom which uses a principal relation
of~$\sign''$.
Indeed, the only GTGDs with a $\sidesign$ atom in their head after
the rewriting are the 
full GTGDs of the form $R'(\vec x)
\rightarrow R(\vec x)$
that we added, and these are full.

\subsection{Homomorphism-closure.} The second and last step is to enforce homomorphism-closure on
the resulting set of GTGDs. We do so simply by considering each GTGD and every
possible way to identify the exported variables and add the resulting GTGD
to~$\Sigma$.

This process does not affect $\sign$-entailment-equivalence, because the
resulting GTGDs are logically entailed by~$\Sigma$. The resulting GTGDs 
also have width no greater than that of the
original GTGDs. Each GTGD still has exactly one principal guard, and each non-full GTGD still has a principal atom in its head.
The process runs
in time polynomial in $|\Sigma|$ and in $2^w$, where $w$ is the width bound.

The resulting set of GTGDs is now homomorphism-closed and the other hypotheses are
still true, so we have finished the proof of Proposition~\ref{prp:strongobey}.

\section{One-pass tree-like chase proofs for guarded TGDs} \label{sec:onepass}

To show our linearization result, we will need a notion of \emph{chase proofs}.
We will more specifically use
\emph{one-pass tree-like chase proofs}: we review the result
from~\cite{gsatvldb} that they can be used for $\owqa$ with GTGDs, and show a
variant of this result that we will use.

\paragraph*{Tree-like chase proofs.}
We first review the general notion of 
\emph{chase sequences} \cite{fagindataex}, which are known to be complete for 
$\owqa$ with CQs and general TGDs.
We specifically focus on \emph{tree-like chase proofs},
which are complete for 
$\owqa$
with guarded TGDs \cite{datalogpmj,baget2010walking}.
We first define an abstract structure of \emph{chase tree}, before clarifying
in the rest of the section how we construct chase trees from an instance and a
set of GTGDs. Our notion of chase trees also distinguishes one node in the
chase tree, which is said to be \emph{recently updated}; we will use the
recently updated node later when defining one-pass chase proofs.

\begin{defi}
A \emph{chase tree} consists of a directed tree $T$, a function mapping each
node $v$ in the tree to a finite set of facts written $T(v)$, and a choice of a node
  of~$T$ called the \emph{recently updated} node. We abuse notation and write
  $T$ to mean both the chase tree and its underlying directed tree.

  Each chase tree has an \emph{underlying instance}, which is just the union of the facts $T(v)$ over all nodes $v$ in the tree.
\end{defi}

We now explain how chase trees can be extended by performing \emph{chase steps} with
GTGDs and so-called \emph{propagation steps}. These steps will then be used to
define \emph{tree-like chase sequences} and their variants.
A chase tree $T$ can be transformed to another chase tree $T'$ in two ways:
\begin{itemize}
    \item One can apply a \emph{chase step} with a GTGD $\gdep: \forall \vec x
      (\body(\vec x) \rightarrow \exists \vec y ~ A(\vec x, \vec y))$. Recall from the preliminaries the definition of triggers and of active triggers.
      Assume that we have a node $v$ in~$T$ and a trigger $\trig$ of~$\gdep$ in~$T$ such that 
      we have ${\trig(\body) \subseteq T(v)}$.%
      Then we can apply a chase step, which we also call \emph{firing}~$\trig$ (on $v$). It will ensure that~$\trig$ is not active in the underlying instance of $T$, as we will add facts to~$T$ that
    define an extension~$\trig'$ of~$\trig$ with an image for the head of~$\gdep$.

    The result of the chase step is obtained as follows.
    \begin{itemize}
      \item If $\gdep$ is full, 
        then 
         the chase tree $T'$ is obtained from $T$ by marking $v$ as recently updated in
         $T'$, setting ${T'(v) \colonequals T(v) \cup \{ \trig(A) \}}$, and defining
         the function $T'$ on other nodes in the same way as~$T$.

        \item If $\gdep$ is not full, then $\trig$ is extended to a
        substitution $\trig'$ that maps each variable in $\vec y$ to a value not
        occurring in $T$, all these values being pairwise distinct. The fresh
        values are domain elements that we call \emph{labelled nulls} or simply 
        \emph{nulls}. The chase tree $T'$ is obtained from $T$ by
        introducing a fresh child $v'$ of $v$, marking $v'$ as recently updated in
        $T'$, and defining $T'$ by extending~$T$ with $T(v')$ which will always contain $\trig'(A)$, and additionally will contain a subset of the following facts of $T(v)$:
\[
{\{ F \in T(v) \mid F
        \text{ is guarded by }
        \trig'(A) \}}
\]
        In other words, the new node contains $\trig'(A)$ and some facts of the
        parent node that are guarded by it, i.e., some facts of $T(v)$ that only
        use elements shared with $T(v')$. We refer to the facts other than $A$
as \emph{inherited facts} of the child node $v'$. 
    \end{itemize}

    \item One can apply a \emph{propagation step} from a node $v$ to a
      node $v'$ in $T$. More precisely, we select a nonempty subset $S \subseteq T(v)$
    of the facts of~$v$, and select a node~$v'$ where 
    these facts do not occur
    (i.e., $T(v') \cap S = \emptyset$) but they are guarded
    (i.e., the facts of $S$ only use elements occurring in a fact
    of~$T(v')$).
    Then we set $T'(v') \colonequals T(v') \cup S$ and mark~$v'$ as
    recently updated.

\end{itemize}
Note that, in the definitions above, we allow the firing of non-active triggers. %
For a non-full TGD firing a non-active trigger means that we create a new
child of~$v$, with a different choice of fresh value for the existentially-quantified variables.

A \emph{tree-like chase sequence} for an
instance $I$ and a finite set of
GTGDs $\Sigma$ 
is a finite sequence of chase trees ${T_0,
\dots, T_n}$. In the sequence, the initial chase tree $T_0$ consists of exactly
one \emph{root node} $r$, with ${T_0(r) \colonequals I}$, and with $r$ being
the recently updated node in $T_0$. Then, each $T_i$ with ${0 < i
\leq n}$ is obtained from~$T_{i-1}$ by one of the two steps above, i.e., a chase
step with some ${\gdep \in
\Sigma}$, or a propagation step. For each node $v$ in $T_n$ and each fact ${F
\in T_n(v)}$, this sequence is a \emph{tree-like chase proof of} $F$
\emph{from} $I$ and $\Sigma$. It is well-known
(e.g., \cite{datalogpmj})
that, for any CQ~$Q$, we have
${I, \Sigma \models Q}$ if
and only if 
there is a tree-like chase sequence $T_0, \ldots, T_n$ for~$I$ and~$\Sigma$ such
that $Q$ has a \emph{match} in~$T_n$, in other words, there is a tree-like chase
proof
$T_0, \ldots, T_n$ of each of the facts of~$\sigma(Q)$, for $\sigma$ some substitution
mapping the variables of~$Q$ to the domain values of $T_n$. Note that the facts
of~$Q$ may be witnessed in different nodes, i.e., it may be the case that there
is no single node $v$ of~$T_n$ such that $\sigma(Q) \subseteq T_n(v)$.

Dating back at least to \cite{datalogpm}, it is known that entailment from GTGDs is witnessed by tree-like chase sequences. We will not make use of results from \cite{datalogpm} directly,
but rather make use of a specialized tree-like chase, based on a construction in \cite{gsatvldbjournal}, which we explain below.

\subsection{Restricted tree-like chase proofs} \label{subsec:onepass}
Our linearization result will rely on the fact that tree-like chase proofs can be
normalized to ensure that we do not jump back and forth in a tree,
but perform our changes to the tree in one single traversal.
Versions of this  result have appeared dating back to~\cite{resultlimitedpods}, with
refinements in~\cite{kevinarxiv} for the disjunctive case.
The result was applied to get rewritings for atomic queries under
GTGD constraints in~\cite{gsatvldb}, and we use the formulation from
that paper.

\begin{defiC}[\cite{gsatvldb}]\label{def:one-pass}
    A tree-like chase sequence ${T_0, \dots, T_n}$ for an
    instance $I$ and
    a finite set of GTGDs $\Sigma$
    is \emph{one-pass} if,
    for each ${0 < i \leq n}$, the chase tree $T_i$ is obtained by applying one of
    the following two steps to  the recently updated node $v$ of $T_{i-1}$:
    \begin{itemize}
        \item  a propagation step copying exactly one fact from
          $v$ to its parent (which then becomes the recently updated node);
        \item  a chase step 
          on $v$ with a GTGD from $\Sigma$ (then either $v$ stays as recently
          updated node or the chase step creates a child of~$v$ which becomes
          the recently updated node).
    \end{itemize}
    Further we can only do the second when the first does not apply.
\end{defiC}

Thus, each chase step in a tree-like chase sequence is applied to a ``focused''
node, namely, an ancestor of the node that was updated by the previous chase
step. Steps with non-full TGDs move the
``focus'' from parent to child, and steps with full TGDs do not move the focus:
the full steps are followed by propagation which copies the fact rootwards as long
as  possible,
possibly moving the ``focus'' rootwards. Moreover, once
a child-to-parent propagation has taken place, the child can never be revisited in
further steps. Indeed, whenever a node stops being the recently updated
node, then the only way it can become recently updated again is following a
propagation step, which always goes rootwards. Hence, if the parent of $v$ node
becomes the recently updated node, then the subtree rooted at~$v$ will never be
revisited again and will no longer be modified.

Theorem~\ref{thm:one-pass-proof-exists}, proven
in \cite{gsatvldbjournal},  shows that one-passness can be enforced on 
chase proofs for GTGDs for fact entailment: whenever a proof exists, there exists a one-pass
proof too.

\begin{thm}[Theorem 4.2 of \cite{gsatvldbjournal}] \label{thm:one-pass-proof-exists}
    For each 
    instance $I$, each finite set of GTGDs $\Sigma$ 
  and each fact $F$ such that ${I, \Sigma \models F}$,
    there exists a one-pass tree-like chase proof of $F$ from $I$ and~$\Sigma$.
\end{thm}

We will need a variant of this one-pass chase process that is aware of the side signature.
We first liberalize the notion of tree-like chase by modifying the definition
of chase steps: we want to also allow the creation of child nodes when
performing a chase step with a full GTGD:
\begin{defi}
A \emph{relaxed tree-like chase} is defined like the notion of tree-like chase defined previously: it
maintains along with each instance a tree
  structure~$T_i$, and designates one node of the tree structure as recently
  updated. The only difference is that we allow \emph{relaxed chase steps with full TGDs}.
  To perform such a step, we consider a full GTGD $\gdep$ having head atom~$A$
and a trigger $\trig$ for~$\gdep$ on some node $v$ in the current tree $T$.
  Performing the relaxed chase step means extending
$T$ to $T'$ by 
        introducing a fresh child $v'$ of $v$, marking $v'$ as recently updated in
        $T'$, and defining $T'$ that extends~$T$ by setting $T(v')$ to
        contain the instantiation $\trig(A)$ of the head atom,
        along with a subset of the facts that are guarded by $\trig(A)$.
That is, in the relaxed step, we  do the same surgery on the tree that we would do on a non-full TGD chase step in a tree-like chase, except that no fresh values are introduced. 

\end{defi}

We then introduce the specific variation that we will use:

\begin{defi}
  Let $\Sigma$ be a finite set of GTGDs
  strongly
    obeying side signature $\sidesign$.
    A \emph{principal-exempt one-pass chase} for an
    instance $I$ and for $\Sigma$
is a relaxed tree-like chase sequence ${T_0, \dots, T_n}$ where
    for each ${0 < i \leq n}$, the chase tree $T_i$ is obtained by applying one of
    the following three steps to the recently-updated node $v$ of $T_{i-1}$:
    \begin{itemize}
        \item  a propagation step copying exactly one side fact from
          $v$ to its parent;
        \item  a
          chase step 
          on $v$ with a GTGD from $\Sigma$, where the GTGD can be either a side
          GTGD or a non-full GTGD;
          \item a 
            relaxed chase step with a full principal TGD.
    \end{itemize}
    Like in the one-pass chase, we further require that we can only perform a chase
    step (of any kind) if no propagation step (of a side fact) applies.
    We also restrict the facts inherited when performing chase steps (of
    either kind) that add a new child node: we require that only side
    facts can be inherited when creating new nodes.
\end{defi}

  It is easy to see that, as $\Sigma$ strongly obeys the side signature, in a
  principal-exempt one-pass chase sequence every non-root node of every tree contains
  precisely one principal fact (which is added when the node is created).
  In particular the definition implies that principal facts are never inherited and also never
  propagated, so each principal fact exists in precisely one node. 
  Intuitively, a principal fact $F$ does not need to be inherited or
  propagated because triggers containing $F$ can always be assumed to
  use $F$ as a guard, and so we will show that they can always be applied on the
  node that contains~$F$.

We now claim a side-signature-aware analogue of
Theorem~\ref{thm:one-pass-proof-exists}: 
for fact entailment, we can always assume that chase proofs
are principal-exempt one-pass.

    \begin{thm}\label{thm:exempt-one-pass-proof-exists}
    For each 
    instance $I_0$, each finite set of GTGDs $\Sigma$ 
      obeying side signature $\sidesign$,
      and each fact $F$ such that ${I_0, \Sigma \models F}$,
    there exists a principal-exempt
      one-pass tree-like chase proof of $F$ from $I_0$ and
    $\Sigma$.
\end{thm}

\begin{proof}
We proceed by reducing to Theorem~\ref{thm:one-pass-proof-exists}. Given the constraints $\Sigma$, we build new constraints $\Sigma'$ over a modification of the signature where
every principal relation $R$ of arity $n$ is replaced by a principal relation $R'$ of arity $n+1$.
For each rule $\sigma \in \Sigma$ we form the rule $\sigma'$ over the revised
  signature by performing the following replacements:
\begin{itemize}
\item Letting $R(\vec x)$ be the principal guard of~$\sigma$, which exists
  because $\Sigma$ strongly obeys~$\sidesign$,
    we replace it by $R'(\vec x, t)$ where $t$ is a fresh variable.
\item If the head of $\sigma$ is a principal atom $R(\vec x)$, then we replace
  it by $R'(\vec x, t')$ where $t'$ is a fresh variable that is existentially quantified.
\end{itemize}

For a fact $F=R(\vec x_0)$ where $R$ is a principal
  relation (i.e., a principal fact), we  let $F'= R'(\vec x_0, t_0)$ where $t_0$ is a fresh constant --
  again, a distinct one for each fact. For a side fact $F$, we let $F'=F$. For
  a finite set of facts $I_0$ we let $I_0'$ be formed by applying this
  transformation to each fact.

  We first claim the following equivalence (*): we have $I_0, \Sigma \models F$
  if and only if $I_0', \Sigma' \models F'$.

  In one direction, consider $I$ extending $I_0$ that satisfies $\Sigma \wedge
  \neg F$. We form $I'$ from $I$ by applying the priming transformation above.
  It is easy to see that $I'$ extends $I_0'$, does not contain $F'$, and also satisfies $\Sigma'$. For the latter, suppose we have a trigger $\tau'$
for $\gdep' \in \Sigma'$ in $I'$. By dropping the extra arguments we get a
  trigger $\tau$ for $\gdep \in \Sigma$ in $I$, thus we have a corresponding
  fact $H$ witnessing the head in $I$, and thus $I'$ witnesses the head of
  $\gdep'$ in $I'$.

  In the other direction, suppose we have $I'$ extending $I_0'$ satisfying
  $\Sigma' \wedge \neg F'$. We form $I$ by simply dropping the last
argument of every principal fact. It is also straightforward that $I$ extends
  $I_0$, does not contain $F$, and satisfies $\Sigma$. For the latter, suppose we have a trigger $\tau$
for $\gdep \in \Sigma$ in $I$. By the definition of $I$, each principal fact can
  be extended in $I'$ with an extra argument. This gives a trigger $\tau'$ for
  $\gdep' \in \Sigma'$ within $I'$. Thus, there is a fact $H'$ witnessing the head of $\gdep'$. We form $H$ by dropping the final argument from $H'$ if it is a principal fact, otherwise we take $H=H'$. This is the fact required to witness that $\tau$ is not active.
This establishes that $I$ satisfies $\Sigma \wedge \neg F$ and proves the
  equivalence (*).

  By Theorem~\ref{thm:one-pass-proof-exists}, entailment of $Q'$ by $I_0' \wedge \Sigma'$ is witnessed by a one-pass proof $T'_1 \ldots T'_n$.
We modify such a proof to $T_1 \ldots T_n$ by simply dropping the final argument in each fact within each node of a tree, each chase step, and
each propagation step. We claim that this can be used to construct a principal-exempt 
  proof of $Q$ from $I_0$ according to $\Sigma$.
The principal-exempt proof will be obtained by performing propagation steps and
  chase steps in the same way as in $T_1 \ldots T_n$, maintaining that after each step the recently updated node in~$T_i$ is the one corresponding to the recently updated node in~$T_i'$.

Let us explain the process more precisely.
Consider the case where $T'_{i+1}$ is formed from $T'_i$ by applying a chase step with a non-full rule $\gdep' \in \Sigma'$ with trigger $\tau'$ on node $v_i'$. There are two cases: either $\gdep'$ corresponds to a non-full rule $\gdep$ of~$\Sigma$, or it corresponds to a full rule $\gdep$ of $\Sigma$ with a principal atom in the head. In the first case, we can fire $\gdep$ on~$\tau$ in~$T_i$ on the node $v_i$ corresponding to~$v_i'$ in~$T_i$, inheriting the facts corresponding to the facts inherited when firing $\tau'_i$ on~$v_i'$. In the second case, we can fire $\gdep$ to perform a relaxed chase step, inheriting the facts corresponding to the facts inherited when firing $\tau'_i$ on~$v_i'$.

Consider now the case where $T_{i+1}'$ is formed from $T_i'$ by applying a chase step with a full rule $\gdep' \in \Sigma'$; this rule corresponds to a full rule $\gdep$ in~$\Sigma$ with a side fact in the head. We fire the corresponding trigger on $T_i$ and create the same fact with a full chase step which has a side atom in the head; and we propagate the fact as much rootwards as it is propagated in $T_i'$. Note that this propagates the newly created facts as much rootwards as possible, because the new fact is guarded by the same nodes in~$T_i$ and in~$T_i'$.

In the sequence $T_1 \ldots T_n$, principal facts are never propagated rootwards, because they are created by firing a trigger which in $T_1' \ldots T_n'$ is a trigger of a non-full rule that creates a fact featuring at least one fresh value: thus the fact is not propagated rootwards in $T_1' \ldots T_n'$ hence not in $T_1 \ldots T_n$. Further, principal facts are never inherited in $T_1 \ldots T_n$, because in $T_1' \ldots T_n'$ these facts feature a fresh value at the extra position,
which is never an exported variable in any rule: so when we create a new child node in $T_1' \ldots T_n'$ then these facts are never guarded by the child node and so can never be inherited, hence the corresponding facts are also not inherited in $T_1 \ldots T_n$.
\end{proof}

\section{Proof overview of the linearization result (Theorem~\ref{thm:idreduce})} \label{sec:proofoverview}

With our normalization of GTGDs and chase proofs out of the way, we can now begin the proof of Theorem~\ref{thm:idreduce}.
The construction that we use is a refinement of the linearization
method given in Section~4.2 of
Gottlob, Manna, and Pieris \cite{gmp}.
Prior to giving the proof, we begin with an  overview.

There is a pretty obvious strategy for linearization: introduce a new predicate
to refine each fact by indicating the set of facts over the side signature
that it guards,
then rewrite the GTGDs into linear TGDs over the enhanced
signature. One also would modify the initial instance accordingly. 
The correctness  of this transformation is less evident. We discuss informally
what the issue is. Entailment with guarded TGDs is captured by the chase process, which forms
an instance with a tree-like shape, as we explained in Section~\ref{sec:onepass}.
At each chase step
we may add new nodes to the tree, but we may also need to propagate facts up and
down the tree. With linear TGDs,
we can use the same process, but \emph{we no longer propagate facts up and down the tree}: we simply
create new nodes containing a single fact per node. This is one of the big advantages of reasoning with linear TGDs over more general guarded TGDs.
It could be the case that, in performing this linearization, we are losing track
of the effect of some propagation steps. We need to use the notion of principal-exempt one-pass chase to
justify that this may not happen (specifically, we will use
Theorem~\ref{thm:exempt-one-pass-proof-exists}).
And we also need to ensure that the
rules are \emph{closed} under a form of logical derivation, to ensure that propagation is unnecessary.

Our proof strategy consists of three steps.

First, we explain how to compute a form of
\emph{saturation} of our GTGDs, which adds full GTGDs that are derived from the original set.
This is a kind of rewriting 
result for atomic query answering
(as in \cite{gsatvldb}). To ensure that the process satisfies our running time
bounds, we 
restrict it to be complete only on a limited subset of instances (the
\emph{childish} instances), and we
do not add \emph{every} derived full GTGDs but restrict to a limited set of
GTGDs, called \emph{suitable}. Because of this restriction,
some care is needed to argue that the resulting saturation is ``complete'' --
for example, that we can derive new facts on an initial 
childish
instance by just
evaluating the full TGDs in the saturation. This completeness is justified using
the principal-exempt one-pass chase: we show that our saturation contains enough
rules to account for the propagation in this version of the chase, which we already know is complete.

This saturation intuitively ensures
that, whenever a full GTGD generates a fact about already-existing elements, then this
generation could already have been performed when these earlier elements had
been generated, using an implied full GTGD in the saturation. 
The main difference of our saturation with~\cite{gmp} is that we exploit the width
and side signature arity bounds to compute only a portion of the derived
rules (the \emph{suitable} ones), without bounding the overall signature.

Second, once this saturation has been
computed, we return to an analysis of the chase and explain how to structure it
further. Specifically, we enforce that we only fire
full GTGDs and their bounded-breadth closure after we have fired a non-full
GTGD. 
We call this the \emph{shortcut chase}, since we shortcut certain derivations that go up and down
the chase tree via the firing of derived rules. We have therefore achieved our
intermediate goal: \emph{a variation of the chase process for
GTGDs that is complete for fact entailment, but with no propagation steps at all}.

Third, using this propagation-free shortcut chase, we can turn to linearization.
We  perform the linearization described informally above -- introduce auxiliary predicates and
then rewrite the source GTGDs to use these predicates. The shortcut chase can then
be used to  argue that the transformation is correct.

The rest of the paper will give the details of this proof template. We fix
the side signature arity bound $a' \in \mathbb{N}$ throughout the proof. Given
the input signature and subsignature, the instance, and the GTGDs, we first
apply Proposition~\ref{prp:strongobey} to compute the signature
$\sign$, the side signature $\sidesign$, the instance $I_0$, and the GTGDs
$\Sigma$ which strongly obey $\sidesign$, while ensuring that $I_0$ and
$\Sigma$ emulate the original instance and GTGDs. 
Remember that this process is polynomial in the input except that it is
exponential in the original width bound, and that the new signature is
polynomial in the original signature.
We let $w$ be the
maximal width of the GTGDs of~$\Sigma$, which is at most the maximum of~$a'$
and of the original width bound.
Up to adding trivial rules to $\Sigma$,
we ensure that the maximal
width~$w$ of a GTGD of $\Sigma$ is such that $w \geq a'$.

\section{Computing a size-controlled saturation} \label{sec:saturation}

We start the first step of our proof by revisiting the well-known notion of \emph{saturation}
of a set of rules. Informally, a saturation of a finite set of GTGDs is a
finite set of \emph{full} GTGDs that derive the same ground consequences as the original
set,
and so is complete for fact entailment.
The notion is closely related to the notion of saturation in resolution theorem-proving, and the notion of a \emph{rewriting} from \cite{gsatvldbjournal}, but our formalization will be slightly different.

The saturation that we will define will not be complete in general instances,
but only on specific kinds of instances, which we call \emph{childish
instances}. These instances are intuitively the ones that can be obtained as the
result of performing a chase step with a bounded-width GTGD, which creates a principal fact and inherits
some side facts. Formally:

\begin{defi}
  \label{def:childish}
    Let $w$ be the maximal width of GTGDs of~$\Sigma$. A
    \emph{childish instance} is an instance $\{F\} \cup I'$ consisting of one principal fact $F
    = R(\vec a)$ which is an isomorphic copy of some GTGD head of~$\Sigma$, together with a set $I'$ of side facts which is guarded by $F$ and where 
    $\dom(I')$ has cardinality at most~$w$.
\end{defi}

Given a finite set of GTGDs $\Sigma$, we say that a finite set of full GTGDs $\Sigma'$ is \emph{complete for fact
entailment} (with~$\Sigma$) if, whenever we apply it to an instance $I$, then all facts on the
domain of~$I$  entailed by~$I$ and~$\Sigma$ are derived by $\Sigma'$.
A \emph{childish saturation} is then a saturation which is complete for fact
entailment on childish instances only. Let us define these notions formally:

\begin{defi} \label{def:childsaturation} A \emph{childish saturation} of a
  finite set of GTGDs $\Sigma$ is a finite set of full GTGDs $\Sigma'$
such that:
\begin{itemize}
\item Every GTGD in $\Sigma'$ is logically entailed by $\Sigma$;
\item $\Sigma'$ is complete for fact entailment over childish instances: for every childish instance $I$, every fact $F$ entailed by $I$ and
  $\Sigma$ is also entailed by $I$ and $\Sigma'$, in other words, $I, \Sigma \models F$
    implies $I, \Sigma' \models F$.
\end{itemize}
\end{defi}

It is known  (see, e.g., \cite{gsatvldbjournal}) %
that every finite set $\Sigma$ of GTGDs has a finite saturation composed of
GTGDs that are complete for fact entailment
on arbitrary finite instances (in particular on childish instances). In fact,
one suitable choice is simply to take all full GTGDs that
are entailed by~$\Sigma$. In this section, we show that,
for dependencies that strongly obey a side
signature, and over childish instances, we can compute a finite saturation that is not too large:
\begin{thm} \label{thm:smallsaturation} There is an algorithm taking as input a
  set of GTGDs $\Sigma$ that strongly obeys
side signature $\sidesign$ and produces a childish saturation, where the maximum width is bounded by the maximum width of $\Sigma$, and the running  time is bounded by 
  \[
    \mathrm{Poly}(|\Sigma|, a^{O(w)}, 2^{n' \times w^{a'}}).
\]
where $a$ is the maximal arity of the relations of~$\sign$, $a'$ is the maximal
  arity of the relations of~$\sidesign$, $n'$ is the number of relations
  of~$\sidesign$, and $w$ is the maximum width of a
  GTGD of~$\Sigma$ (assumed to be no smaller than~$a'$).
\end{thm}

We prove this theorem in the rest of this section.
We will define the \emph{$\sidesign$-suitable saturation} of $\Sigma$, and show that it has the properties required by the theorem.
We will reason about full GTGDs with
side signature $\sidesign$ that are \emph{suitable}, i.e., that obey 
three requirements: having width at most $w$, being \emph{$\Sigma$-compatible},
and satisfying a certain \emph{breadth} restriction.
The notion of 
\emph{$\Sigma$-compatibility} means that the head atoms and principal
guard atoms of GTGDs are \emph{compatible} with~$\Sigma$ in the sense that they
are isomorphic to some principal head atom of~$\Sigma$.
Restricting to such GTGDs is necessary to avoid considering a number of GTGDs
which would be exponential in the
signature arity.
As for \emph{breadth},
it bounds how many different variables can be
used by the side atoms of any given GTGD, again avoiding an exponential blowup
in the number of possible GTGDs as a function of the principal signature arity.
The notion of breadth intuitively means that the body of a
suitable full GTGD is a childish instance, up to replacing variables by domain
elements.
Formally:

\begin{defi}
  We say an atom $A$ is \emph{$\Sigma$-compatible} if it is a side signature
  atom or if there is 
  a head atom in a GTGD of~$\Sigma$ to which it is isomorphic.

  Let $\gdep$ be a full GTGD obeying side signature~$\sidesign$ having width at
  most~$w$ and having a principal guard $B_\gdep$. Let $H_\gdep$ be its head
  atom. We say that $\gdep$ is \emph{$\Sigma$-compatible} if each one of
  $H_\gdep$ and $B_\gdep$ are $\Sigma$-compatible (not necessarily with the
  same atom of~$\Sigma$).
\end{defi}

\begin{defi}
  For any $b \in \mathbb{N}$, 
  we say that $\gdep$ has \emph{breadth $\leq b$} if, letting $A$ be its
  principal guard, then there is a subset~$X$ of at most $b$ variables of~$A$ such that the
  other atoms of the body of~$\gdep$ only use variables of~$X$.

 \end{defi}

 The formal definition of \emph{suitable GTGDs} is then:

  \begin{defi}
    Letting $w$ be the maximal width of GTGDs of~$\Sigma$, a
  \emph{$\Sigma$-suitable GTGD} is a GTGD which is full, is
    $\Sigma$-compatible,
    has exactly one principal guard,
  has breadth at most~$w$,
    and has width at most~$w$.
    When $\Sigma$ is clear from context, we refer simply to a \emph{suitable
    GTGD}.
\end{defi}

\begin{exa}
  Let $n \in \mathbb{N}$ be an integer, let the signature consist of a
  principal relation
  $R$ of arity $n$ and of a single binary relation $S$ for the side signature. For
  $w = 2$, the following full GTGD has width~$\leq w$ but is not suitable
  (because it does not have breadth $\leq w$):
  \[
    R(x_1 \ldots x_n), S(x_1, x_2), \ldots, S(x_{n-2}, x_{n-1}) \rightarrow
    S(x_{n-1}, x_n)
  \]
\end{exa}

Note that the GTGDs of $\Sigma$, even the full GTGDs of~$\Sigma$, may not all be
suitable because they do not satisfy the
breadth bound. Intuitively, the non-suitable full GTGDs of~$\Sigma$ will still be considered in
the saturation process, but the process will only create new full GTGDs that are
suitable.

We compute a bound on the number of suitable full GTGDs, as a function of the
arity and size of the signature and of the side signature, together with the
width and the size of the set $\Sigma$ of GTGDs. This uses
the fact that the suitable full GTGDs are $\Sigma$-compatible, have bounded
width, and have bounded breadth:

\begin{lem}
  \label{lem:numbersuitable}
  The number of suitable full GTGDs is at most:
  \[
    |\Sigma|^2 \times (a+1)^{3 w} \times 2^{n' \times w^{a'}}
  \]
 where: 
  \begin{itemize}
    \item $|\Sigma|$ is the number of GTGDs in~$\Sigma$,
    \item $a$ is the maximal arity of any relation in~$\sign$,
    \item $n'$ is the number of relations in the side signature~$\sidesign$,
    \item $a'$ is the maximal arity of the relations of~$\sidesign$,
    \item $w$ is the maximal width of a GTGD of~$\Sigma$.
  \end{itemize}
\end{lem}

\begin{proof}
  We construct a suitable full GTGD by:
  \begin{itemize}
    \item Picking a principal guard atom $A$ which is isomorphic to a head atom
      of~$\Sigma$: this gives $|\Sigma|$ choices, and the resulting atom has at
      most $a$ variables.
    \item Picking a subset of variables of~$A$
      on which to add side facts: by the breadth bound we pick at most~$w$ of the~$a$
      variables, so the number of choices can be overapproximated as
      $(a+1)^{w}$.
    \item Picking an instance of side facts on a domain of size at
      most~$w$:
      \begin{itemize}
        \item Each possible fact is obtained by picking a relation (among $n'$),
          and filling every position (of which there are at most~$a'$) with an
          element (of which there are at most~$w$), i.e., there are at most $n'
          \times w^{a'}$ possible facts.
        \item So, for the choice of sets of side facts, we have $2^{n'
          \times w^{a'}}$ options.
      \end{itemize}
    \item Picking a head atom $H$ which is isomorphic to a head atom
      of~$\Sigma$: this gives at most $|\Sigma|$ choices, and again the head
      atom has at most~$a$ variables.
    \item Picking a sequence of exported variables from the body atom: this can
      be overapproximated as $(a+1)^{w}$ possible sequences.
    \item Picking a sequence of exported variables from the head atom: again
      $(a+1)^{w}$ possible sequences. The other variables are existentially
      quantified.
  \end{itemize}
  Putting it together, we obtain the claimed bound.
\end{proof}

Observe that, when $w$, $n'$, and $a'$ are all constant, then the above quantity
is polynomial in the size of the input signature~$\sign$.
Further, when only $a'$
is bounded, then the quantity is singly exponential in the input.

Our goal in focusing on suitable full GTGDs is to identify which ones are
\emph{derived}, i.e., follow from~$\Sigma$. We say that a suitable
full GTGD $\gdep$ is 
a \emph{derived suitable full GTGD} if we have
 $\Sigma \models \gdep$, 
that is, any instance
that satisfies $\Sigma$ also satisfies $\gdep$. Again, the
derived suitable GTGDs generally do \emph{not} include all the full GTGDs
of~$\Sigma$, because not all of them are suitable.

For our proof of 
Theorem~\ref{thm:smallsaturation}, we will need to show that the set of derived suitable full GTGDs can be
computed efficiently. Let us define a specific kind of derived suitable full
GTGD, namely, the \emph{trivial} ones:

\begin{defi} \label{def:suitable}
  We say that a full GTGD $\gdep$ is \emph{trivial} if its head atom
  already occurs in its body. A \emph{trivial suitable} full GTGD is a full GTGD
  which is both trivial and suitable.

  We also define from $\Sigma$ the set $\Sigma_\triv$ of trivial full GTGDs
  where the body contains a principal atom $A$ which is an isomorphic copy of a
  GTGD head of~$\Sigma$, the other atoms of the body are side atoms on at most
  $w$ different variables of~$A$, and the head is identical to~$A$. The full
  GTGDs of~$\Sigma_\triv$ are all trivial, they are all $\Sigma$-compatible, and
  they all satisfy the breadth bound: but they do generally not satisfy the
  width bound.
\end{defi}

Note that not all trivial GTGDs are suitable, because even trivial GTGDs may
have unbounded width and unbounded breadth. This is why we will restrict to derived
suitable trivial GTGDs (to be put in our saturation), and the GTGDs of
$\Sigma_\triv$ (which will be considered in the saturation process but will not be part
of the saturation, because they are not suitable).

We now define the \emph{saturation}, which is computed by starting with the
suitable trivial full GTGDs and the suitable full GTGDs of~$\Sigma$, and
closing under the application of two
rules, (Transitivity) and (Principal+Transitivity). Intuitively, (Transitivity)
can be used to deduce new full GTGDs by combining full GTGDs with those already
deduced; and (Principal+Transitivity) has the same purpose but where we
additionally perform a rewriting of one of the GTGDs via a principal GTGD.
(Recall from Definition \ref{defi:principal} that a principal GTGD is a GTGD whose head is a principal relation.)

\begin{defi}
  \label{def:satur}
  Given a set of GTGDs $\Sigma$ that strongly obey side signature
  $\sidesign$, the
  \emph{$\sidesign$-suitable saturation} $\wclo{\Sigma}$ is obtained by starting with
  the suitable full GTGDs in~$\Sigma$, plus the
  trivial suitable full GTGDs,
  and applying the following inference rules until we
  reach a fixpoint:
\begin{itemize}
\item (Transitivity): 
  Suppose that $\wclo{\Sigma} \cup \Sigma_{\triv}$ contains $n$ full GTGDs with the same
    body (up to renaming), that is, it contains full GTGDs $\beta \rightarrow
    B_1(\vec z_1), \ldots, \beta \rightarrow B_n(\vec z_n)$. 
  
    Suppose that there is a full GTGD $\beta' \rightarrow \rho'$ 
    in~$\Sigma \cup \wclo{\Sigma}$, and 
    that there is a homomorphism $\upsilon$ mapping 
    $\beta'$ to~$\beta \wedge \bigwedge_j B_j(\vec z_j)$.
    Then add to~$\wclo{\Sigma}$ the following if it is suitable:
    \[ \beta \rightarrow \upsilon(\rho').
\]
\item (Principal+Transitivity)
    Suppose that $\wclo{\Sigma} \cup \Sigma_{\triv}$ contains $n$ full GTGDs with the same
    body (up to renaming), that is, it contains full GTGDs $\beta \rightarrow
    B_1(\vec z_1), \ldots, \beta \rightarrow B_n(\vec z_n)$. 
    Let $\dep_{cc}$ be a principal GTGD of~$\Sigma$. We use the subscript $cc$ to emphasize that the dependency ``creates a child'' in the chase. 
    Let $A_{cc}$ be its principal guard,
    let $\beta_{cc}$ be the conjunction of its side atoms,
    and let $H_{cc}$ be its head.
    Let $\dep'$ be a full GTGD of $\wclo{\Sigma}$ 
    whose principal guard $A'$ is isomorphic to~$H_{cc}$: up to renaming the
    variables of~$\dep'$ we assume that $A'$ is identical to~$H_{cc}$.
    Let $\beta'$ be the
    conjunction of the side atoms of~$\dep'$ and let $H'$ be its head atom. 
    Assume that $\beta'$ and $H'$ only use variables of~$H_{cc}$ that are
    exported variables of~$\dep_{cc}$. 
    Further assume that $A_{cc} \land \beta_{cc} \land \beta'$ can
    be mapped by a homomorphism $v$ to $\beta \land \bigwedge_j B_j(\vec z_j)$. Then 
    add the following full GTGD to~$\wclo{\Sigma}$ if it is
    suitable:
    \[
      \beta \rightarrow v(H')
    \]
\end{itemize}
\end{defi}

Observe that (Principal+Transitivity) is quite similar to (Transitivity), but
intuitively we are additionally composing with a principal GTGD $\dep_{cc}$ of $\Sigma$.
We also note that a similar saturation process in discussed in~\cite{resultlimitedj}, in
the specific case of accessibility axioms for access methods, namely in the
proof of \cite[Proposition 6.6]{resultlimitedj}.
Other algorithms for ``Datalog rewriting'' of GTGDs have used similar closure rules -- e.g. the FullDR algorithm of \cite{gsatvldbjournal}. As in those cases,  one has a closure rule that composes full GTGDs --in our case (Transitivity). And one also has a closure rule -- in this case, (Principal+Transitivity) -- that composes full GTGDs and another GTGD that may be non-full, provided
that one can compose in a way that gives a full GTGD. 

One rough intuition for these rules comes from arguing inductively that
full GTGDs should suffice to capture the generation of facts
on a given node $v$ in a principal-exempt one-pass chase.
One way that a fact $F$ can come into node $v$ is that we create a child $c$ of $v$ in the
chase, using a principal GTGD $\delta_{cc}$ in the original set $\Sigma$, then
generate a fact $F$ in $c$, and propagate $F$ rootwards back to $v$.
We can break up the chase into three parts: (1.) generating the facts in $v$
required to fire $\delta_{cc}$ and the other facts which will be inherited
in~$c$ and used to generate~$F$; (2.) the firing of $\delta_{cc}$; and (3.) a
chase sequence on $c$ that generates~$F$.
Step (1.) will intuitively be inductively captured by the derived dependencies $\beta \rightarrow
    B_1(\vec z_1), \ldots, \beta \rightarrow B_n(\vec z_n)$; step (2.)
    corresponds to $\delta_{cc}$; and step (3.) corresponds to the full GTGD $\delta'$.
    The (Principal+Transitivity) inference rule is used to create
    a full GTGD that captures the composition of the whole process and can be
    applied directly on node~$v$.
    We point the reader to Case~2 in the proof of Claim~\ref{clm:derived} within
    the completeness argument, where this is explained more formally.

\begin{exa}
  \label{exa:closure}
  We first illustrate the inference rule (Transitivity). Assume that $\Sigma \cup
  \wclo{\Sigma}$ contains the following full GTGDs:
  \begin{align*}
    R(x, y_1, \ldots, y_n, z), S(x) & \rightarrow T(y_1)\\
    & \vdots \\
    R(x, y_1, \ldots, y_n, z), S(x) & \rightarrow T(y_n)\\
    R(x, y_1, \ldots, y_n, z), T(y_1), \ldots, T(y_n) & \rightarrow U(z)
  \end{align*}
  Note that these GTGDs obey the side signature $\{S, T, U\}$ and have width at
  most 1; but the last one does not have bounded breadth. The (Transitivity)
  inference rule allows us to deduce the following full GTGD, which has width 1 and
  breadth 1:
  \[
    R(x, y_1, \ldots, y_n, z), S(x) \rightarrow U(z)
  \]

  We now illustrate the inference rule (Principal+Transitivity). Assume that the principal signature contains a $4$-ary relation $R$ and a ternary relation $R'$, and that all other relations are in the side signature.
  Assume that $\wclo{\Sigma}$ contains the following full GTGDs:
  \begin{align*}
    \gdep_1: R(x_1, x_2, y_1, y_2), S(x_1), S(x_2) & \rightarrow T(y_1)\\
    \gdep_2 : R(x_1, x_2, y_1, y_2), S(x_1), S(x_2) & \rightarrow T(y_2)
  \end{align*}
  And it also contains the full GTGD
  \begin{align*}
     \gdep_{3} : R'(y_1, y_2, z), T(y_1), T(y_2) & \rightarrow U(y_1, y_2)
     \end{align*}
  Note that these  GTGDs have breadth $2$ and width at most~$2$. Assume that $\Sigma$ contains the following non-full GTGD of width $2$ and breadth $2$:
  \[
  \gdep_{cc}: R(x_1, x_2, y_1, y_2), S(x_1), S(x_2) \rightarrow \exists z' R'(y_1, y_2, z)
  \]
  Applying the inference rule (Principal+Transitivity),   we deduce:
  \[
    R(x_1, x_2, y_1, y_2), S(x_1), S(x_2) \rightarrow U(y_1, y_2)
  \]
This captures the effect of applying GTGDs $\gdep_1$ and $\gdep_2$ to get the
  additional $T$-facts, then the non-full GTGD $\gdep_{cc}$ to get the $R'$-fact
  (which guards the $T$-facts), and finally the GTGD $\gdep_{3}$.

  To understand why the inference rule is written the way it is, note that a na\"ive composition of  $\gdep_{3}$ via $\gdep_{cc}$ would have given:
  \[
    R(x_1, x_2, y_1, y_2), S(x_1), S(x_2), T(y_1), T(y_2) \rightarrow U(y_1, y_2)
  \]

  And the inference of such a rule would indeed have been generated in a saturation algorithm that was not concerned with the size.
  But this rule would have breadth $4$. For this reason, the
  (Principal+Transitivity) inference rule intuitively performs a step analogous to (Transitivity) but via a principal GTGD, all in one go.

  Further note that the inference rule (Principal+Transitivity) also applies for principal full GTGDs of~$\Sigma$. For example, consider the modification of the last case above, where $R'$ is a binary relation without its last position. Then $\gdep$ is a full GTGD, and (Principal+Transitivity) can be applied in the same way.

  We finish by exemplifying the purpose of considering the rules of
  $\Sigma_{\triv}$, in the case of (Principal+Transitivity). Consider the following GTGDs, where the first is non-full and
  the second is full:
  \begin{align*}
    \gdep_1': R(x_1 \ldots x_n), U_1(x_n) & \rightarrow \exists x_{n+1} \cdots x_{2n} S(x_n \ldots x_{2n})\\
    \gdep_2': S(x_n \ldots x_{2n}), U_2(x_n) & \rightarrow U_3(x_n)
  \end{align*}
  where $U_1$ and $U_2$ and $U_3$ are unary side relations and the other
  relations are principal. The following is a derived suitable full GTGD:
  \[
    \gdep_3': R(x_1 \ldots x_n), U_1(x_n), U_2(x_n) \rightarrow U_3(x_n)
  \]
  To derive it and add it to $\wclo{\Sigma}$, we use (Principal+Transitivity)
  with the following trivial full GTGD
  of~$\Sigma_{\triv}$:
  \[
    \gdep': R(x_1 \ldots x_n), U_1(x_n), U_2(x_n) \rightarrow R(x_1 \ldots x_n)
  \]
  This GTGD is not suitable, because it does not satisfy the width bound; but it
  is needed to give us a way to have $\beta = \{R(x_1 \ldots x_n), U_1(x_n),
  U_2(x_n)\}$.

\end{exa}

Clearly, by definition, all the GTGDs of $\wclo{\Sigma}$ are suitable.
We now state and prove that the computation of the
saturation can be performed efficiently:

\begin{lem}
  \label{lem:computeclosure}
  There is a polynomial $P$ such that,
  for any set $\Sigma$ of 
  GTGDs of width at most~$w$
  which strongly
  obey side signature $\sidesign$,
  letting $a$ be the arity of the signature, $n'$ the number of relations in the
  side signature, and
  $a'$ the arity of the side signature,
  we can compute
  $\wclo{\Sigma}$ in time
  $P(|\Sigma|, a^{O(w)}, 2^{n' \times w^{a'}})$.
\end{lem}

\begin{proof}
  From our bound in Lemma~\ref{lem:numbersuitable}, and knowing that the full GTGDs
  in the saturation are suitable,
  we know that the maximal size of~$\wclo{\Sigma}$ satisfies our running time
  bound. 
  We also know by an immediate variant of Lemma~\ref{lem:numbersuitable} that
  the number of GTGDs in $\Sigma_{\triv}$ satisfies the running time bound,
  because GTGDs of~$\Sigma_{\triv}$ can be obtained by picking the body of a
  suitable GTGD, and then picking a head which is identical to the guard atom.
  So in the algorithm below, we materialize $\Sigma_{\triv}$.

  We can compute $\wclo{\Sigma}$ by iterating the possible production of rules until
  we reach a fixpoint, so it suffices to show
  that at each intermediate state of~$\wclo{\Sigma}$,
  testing every possible inference rule application is in $\ptime$ in 
  $|\wclo{\Sigma} \cup \Sigma|$ for $\wclo{\Sigma}$ the set of derived
  suitable full GTGDs that have currently been computed.

  We first explain how to test for applications of the  (Transitivity)
  inference rule.
  We first make the bodies of the rules of~$\wclo{\Sigma}$ canonical
  by first giving canonical names to the
  variables of the principal guard atom, in the order in which they appear, and
  then sorting the side signature atoms in some canonical way (e.g., by
  relation, then by lexicographic order on the variables). This can be
  done in polynomial time, and makes it easy to
  regroup all the GTGDs of~$\wclo{\Sigma}$ that have an isomorphic body, and
  then try each possible body as a choice of~$\beta$.

  Now, for each choice of body~$\beta$, we can consider all GTGDs
$\beta \rightarrow
    B_1(\vec z_1), \ldots, \beta \rightarrow B_n(\vec z_n)$
    of $\wclo{\Sigma}$ having body~$\beta$,
  and consider the union~$H$ of their
  heads. Indeed, note that when applying (Transitivity) we can always assume
  without loss of generality that we consider all full GTGDs of~$\wclo{\Sigma}$
  having the body~$\beta$, because having more such GTGDs will give us more rule
  heads $B_j(\vec z_j)$, which makes it easier to apply (Transitivity) to some
  choice of GTGD $\beta' \rightarrow \rho$. 
  Now, we enumerate all GTGDs in~$\Sigma \cup \wclo{\Sigma}$ and we
  test whether their body $\beta'$ can be mapped homomorphically to
  $\beta \cup H$, and apply the inference rule (Transitivity) if that is the case. We must argue that this test can be done in $\ptime$.
  We do this by considering each principal atom $A$ in~$\beta \cup H$,
  testing in $\ptime$ for each of them whether 
  the principal guard of~$\beta'$ can be mapped homomorphically to~$A$,
  and see whether the
  mapping thus defined is a homomorphism from~$\beta'$ to~$\beta\cup H$.
  If this is the case, 
  we can build the new full GTGD
  and add it to~$\wclo{\Sigma}$ if it is suitable (which is easy to verify).

  We now explain how to test for possible applications of the
  (Principal+Transitivity) inference rule.
  We make the GTGD bodies of~$\wclo{\Sigma} \cup \Sigma_{\triv}$ canonical like in the previous
  case: consider every
  GTGD body $\beta$ in~$\wclo{\Sigma} \cup \Sigma_{\triv}$, and for each~$\beta$ consider all GTGDs
  of~$\wclo{\Sigma}$ having $\beta$ as body.
  Now, we consider every choice of a principal GTGD $\dep_{cc}$
  of~$\Sigma$ and a full GTGD $\dep'$ of~$\wclo{\Sigma}$. We test if these
  satisfy the conditions to apply the inference rule (Principal+Transitivity). Reusing the
  notations used in the inference rule,
  we check that the principal guard~$A'$ of~$\gdep'$ is isomorphic to $H_{cc}$: if
it is not then we cannot apply the inference rule for this choice of
  $\dep_{cc}$ and $\dep'$; if is is
  isomorphic then this defines the isomorphism~$\iota$ and we continue
  with trying to apply the inference rule.
  We now define $\beta''$ and $H''$ and check that $A_{cc} \land \beta_2 \land
  \beta''$ maps homomorphically to~$\beta \land \bigwedge_j B_j(\vec z_j)$ by defining first the
  homomorphism from $A_{cc}$ to the principal atom of~$\beta$, and checking if this
  mapping is a homomorphism from $A_{cc} \land \beta_2 \land
  \beta''$ to~$\beta \land \bigwedge_j B_j(\vec z_j)$.
  If this check succeeds, we build the new full GTGD and add it to~$\wclo{\Sigma}$ if it is suitable. As noted earlier, suitability is easy to verify. Again, this can all be performed in $\ptime$, since we are dealing with sets of ground atoms guarded by a single (unique) principal ground atom, and thus there is only one possible isomorphism.

  Hence, we can compute the process by testing every possible inference rule application
  on the current $\wclo{\Sigma}$, building the possible new GTGDs, and adding them
  to~$\wclo{\Sigma}$ if they are new and suitable. This can be done in polynomial time in the dependencies $\Sigma$ and in the current size of the set $\wclo{\Sigma}$.
  We can continue this process
  as long as the size of~$\wclo{\Sigma}$ increases. From the bound on the
  maximal size of~$\wclo{\Sigma}$, we conclude that the running time bound is
  respected.
\end{proof}

Let us now conclude the proof of Theorem~\ref{thm:smallsaturation} by showing
that the full GTGDs of~$\Sigma$, together with $\wclo{\Sigma}$, are actually a
childish saturation in the sense of Definition~\ref{def:childsaturation}. Towards showing this, we first establish that $\wclo{\Sigma}$ contains all derived suitable full GTGDs:

\begin{clm} \label{clm:derived}
  The set of GTGDs $\wclo{\Sigma}$ is exactly
  the set of derived suitable full GTGDs.
\end{clm}

\begin{proof}
  It is clear by definition that the full GTGDs in $\wclo{\Sigma}$ are all suitable,
  and an immediate induction shows that they are all derived (i.e., they all
  logically follow from~$\Sigma$). %
  So let us focus on the converse:
  let us consider a derived suitable full GTGD $\gdep: A
  \land \beta \rightarrow H$, and show that it is in $\wclo{\Sigma}$.
  From the width bound on~$\gdep \in \Sigma$, since $\gdep$ is full, we know
  that $H$ has at most~$w$ different
  variables.

  As $\gdep$ is derived, let us pick any 
  isomorphism $\iota$ to transform the principal guard $A$, the side atoms
  $\beta$, and the head $H$
  respectively into:
  a principal fact $A_0 := \iota(A)$, a set of side facts $\beta_0 :=
  \iota(\beta)$, and a fact
  $H_0 := \iota(H)$.  Let us consider a proof of $H_0$ from $\inst_0 = \{A_0\}
  \cup \beta_0$. 
  Using Theorem~\ref{thm:exempt-one-pass-proof-exists}
  let us more precisely take a principal-exempt
  one-pass chase proof of~$H_0$
  from~$\inst_0$ using the GTGDs of $\Sigma$.
  As $\Sigma$ is homomorphism-closed, we can assume without loss of generality that, in all
  GTGD firings, no two distinct exported variables are mapped to the
  same domain element.

  We show that $\gdep$ is in $\wclo{\Sigma}$ by
  induction on the length of such a principal-exempt 
  one-pass chase derivation
  of~$H_0$ from~$\inst_0$ using~$\Sigma$. The
  base case corresponds to the case of an empty derivation, in which case $H_0
  \in \inst_0$ so that $\gdep$ is a trivial GTGD and we immediately conclude because
  all trivial suitable full GTGDs are in~$\wclo{\Sigma}$ by definition. %

  To show the induction step, assume without loss of generality that the
  principal-exempt 
  one-pass chase proof of~$H_0$ from~$\inst_0$ finishes by deriving
  $H_0$, and let $v$ be the node on which the last chase step is performed to
  derive~$H_0$. We distinguish two cases,
  depending on whether $v$ is the root node $v_0$ of the
  tree-like chase sequence, or whether $v$ is a strict
  descendant of~$v_0$.

  \medskip
  \textbf{Case 1: the last firing is performed on the root~$v_0$.}
  If $v = v_0$, then the last chase step fired a GTGD $\gdep'$
  of~$\Sigma$ on a trigger~$\trig$. The GTGD $\gdep'$  must be a full GTGD because it derived the fact $H_0$,
  which contained only values from~$v_0$.
  In the trigger~$\trig$, the principal guard
  of~$\gdep'$ was mapped to the principal guard~$A_0$ of~$v_0$. We can assume
  by homomorphism-closure that the mapping
  is a bijection on the exported variables.
  If $\gdep'$ has no side atoms,
  then we conclude immediately by applying (Transitivity) with
  the full trivial GTGD $A \land \beta \rightarrow A$ of $\Sigma_{\triv}$, 
  together with $\gdep'$. So in what follows we assume that $\gdep'$ has some
  side atoms.

  The side atoms of~$\gdep'$ were mapped to some
  facts $\beta'_0$ of~$v_0$ which were either part of~$\beta_0$ or were derived
  earlier by the chase sequence.
  In other words, for each fact $F$ of~$\beta'_0$, there is a principal-exempt one-pass chase proof of~$F$
  from~$\inst_0$ (possibly of length~$0$), which witnesses that the full GTGD $\gdep_F: A \land \beta \rightarrow
  \iota^{-1}(F)$ is a logical consequence of~$\Sigma$.
  Now, we can argue that, for each~$F$, the GTGD $\gdep_F$ is
  suitable.  It is obviously a full GTGD. In terms of the compatibility requirements, it has the same body as the suitable full GTGD $\gdep$ so 
it satisfies the
  breadth bound.
  It has
  exactly one principal guard, the principal guard is $\Sigma$-compatible, and its head is a side atom so it is also  $\Sigma$-compatible. 
  Finally, its head atom is for a side signature relation, 
  so the width of $\gdep_F$ is at most the arity of that atom, i.e., at most 
  the side signature arity.
  
  Further, $\gdep_F$ is a derived suitable full GTGD, which has 
  a strictly shorter principal-exempt 
  one-pass chase
  proof from~$\inst_0$ and~$\Sigma$.
  Thus, by the induction hypothesis,
  $\gdep_F \in \wclo{\Sigma}$. Now, applying
  (Transitivity) to the full GTGDs $\gdep_F$ of~$\wclo{\Sigma}$
  for the non-empty set $F \in \beta_0'$ and to
  the full GTGD $\gdep'$ of~$\Sigma$, we conclude that our initial full GTGD
  $\gdep$, which is suitable by hypothesis, was in~$\wclo{\Sigma}$.

  \medskip
  \textbf{Case 2: the last firing is performed on a strict descendant of the root node ~$v_0$.}
  If the last firing was performed on a node $v \neq v_0$, then 
  $v$ is a strict descendant of~$v_0$.
  Let $v_1$ be the child of~$v_0$ which is an
  ancestor of~$v$; possibly $v = v_1$. Let $A_1$ be the principal guard of~$v_1$, and let $\beta_1$
  be the facts inherited from $v_0$ when creating~$v_1$: 
  by definition of the principal-exempt one-pass chase, the facts of~$\beta_1$
  are all side facts.
  Let $\gdep_{cc}$ be the GTGD of~$\Sigma$ (full or non-full)
  which was fired earlier to create node~$v_1$. By definition of the
  principal-exempt 
  one-pass chase, $\gdep_{cc}$ is a principal GTGD, and its head 
  was instantiated to~$A_1$. Let $\trig_{cc}$ be the trigger fired when
  creating~$v_1$. We know that $\trig_{cc}$ consists of the principal guard~$A_0$ of
  the root~$v_0$ 
  together with some side facts $\beta_0'$ which were present in~$v_0$ at that moment. And
  thanks to homomorphism-closure we know the following fact (*): in~$\trig_{cc}$ 
  no two distinct exported variables are mapped to the same element.
  Let us rename the exported variables of $\gdep_{cc}$ to match the variables of the
  guard $A$ of~$\gdep$, and let us rename the existentially quantified variables
  of~$\gdep_{cc}$  to use fresh variable names.
  We can extend the isomorphism $\iota$, which was originally defined on the
  variables of~$A$, and extend it to an isomorphism defined on the variables of
  $A$ and also on the head of~$\gdep_{cc}$, which is mapped 
  to the fact~$A_1$.
  We can then split the chase sequence into three successive parts:
  \begin{itemize}
    \item The initial part, which starts with $\inst_0$ and creates the facts of $\beta_0'$
      in~$v_0$ in some order.
      For each fact
      $F \in \beta_0' \setminus \beta_0$, we now reason as in Case 1 above
      to show that the full GTGD $\gdep_F: A \land \beta
      \rightarrow \iota^{-1}(\beta_0')$ is suitable and is in
      $\wclo{\Sigma}$: the latter uses the induction hypothesis.
    \item The firing of~$\trig_{cc}$ for the principal GTGD $\gdep_{cc}$ of~$\Sigma$, which
      creates~$v_1$ containing the head instantiation $A_1$ along with some facts
      $\beta_1 \subseteq \beta_0'$ inherited from~$v_0$.
    \item The subsequent part of the chase sequence, which is performed in the subtree
      rooted at~$v_1$ by definition of the one-pass chase, and which derives the fact~$H_0$ from
      $I_1 = \{A_1\} \cup \beta_1$. Note that $H_0$ must be on elements
      shared between~$A_0$ and~$A_1$
  \end{itemize}
  Now, let us consider the full GTGD $\gdep'' : \iota^{-1}(A_1) \land \iota^{-1}(\beta_1)
  \rightarrow H$. The third bullet point above witnesses that this
  full GTGD is a logical consequence of~$\Sigma$. Let us show that it is suitable. It
  has precisely one principal guard.  It obeys the width bound because $H$ is
  the head of~$\gdep$ which is suitable by assumption so it has at most $w$ different variables.
  It obeys the breadth bound because
  $\beta_1$ uses at most $w$ different elements of~$A_1$ thanks to the width
  bound on~$\gdep_{cc}$. Now let us verify the compatibility conditions. The principal guard is $\Sigma$-compatible because $A_1$
  was created by instantiating the head of~$\gdep_{cc}$ (without identifying any
  variables thanks to (*)), so $\iota^{-1}(A_1)$ is actually exactly the head
  of~$\gdep_{cc}$. And the head $H$ is the head of
  $\gdep$, and since $\gdep$ is suitable its head is $\Sigma$-compatible. Thus, $\gdep''$ is a
  suitable derived full GTGD with a strictly shorter chase proof. By induction hypothesis, $\gdep'' \in \wclo{\Sigma}$.

  We now distinguish two subcases: either $\beta'_0 \setminus \beta_0$ is empty,
  or it is non-empty. We first deal with the subcase where it is empty. Then we
  want to apply the inference rule (Principal+Transitivity) with $\gdep_{cc} \in
  \Sigma$, $\gdep'' \in \wclo{\Sigma}$, and the trivial full GTGD of
  $\Sigma_{\triv}$ with same body as~$\gdep$, namely: $A \land \beta
  \rightarrow A$. Let us show that the inference rule is indeed applicable. We know that
  the principal guard of~$\gdep''$ is exactly the head of~$\gdep_{cc}$. The
  facts of $\beta_1$ were inherited from~$v_0$, and since $\beta'_0 \setminus
  \beta_0$ is empty they were all part of~$\beta_0$. So the side
  atoms $\iota^{-1}(\beta_1)$ of~$\gdep''$, together with the body of
  $\gdep_{cc}$, can be homomorphically mapped to $A \land \beta$, as required to
  apply the rule.

  In the subcase where $\beta'_0 \setminus \beta_0$ is non-empty, we want to
  apply the inference rule (Principal+Transitivity) again with $\gdep_{cc} \in \Sigma$
  and $\gdep'' \in \wclo{\Sigma}$, but this time together with  
  the non-empty set of the $\gdep_F$ in~$\wclo{\Sigma}$ for $F \in \beta'_0
  \setminus \beta_0$ instead of using a trivial full GTGD of~$\Sigma_{\triv}$.
  Let us show that the inference rule is indeed applicable. We already know that
  the principal guard $\iota^{-1}(A_1)$ of $\gdep''$ is exactly the head
  of~$\gamma_{cc}$.
  The facts of $\beta_1$ were inherited from~$v_0$, and we have
  explained that they were part of~$\beta_0$ or were created
  by instantiating the heads of the GTGDs $\gdep_F$. Further, the body
  of~$\gdep_{cc}$ can be mapped to $\beta_0'$, so these facts were also part
  of~$\beta_0$ or were created by instantiating the heads of the GTGDs
  $\gdep_F$.
  Thus, the side atoms $\iota^{-1}(\beta_1)$
  of~$\gdep''$, together with the body of
  $\gdep_{cc}$, can be homomorphically mapped to the conjunction of
  $A \land \beta$ and of the heads of the GTGDs $\gdep_F$, as required to
  apply the rule.

  The application of (Principal+Transitivity) deduces $A \land \beta \rightarrow H$,
  namely our initial full GTGD
  $\gdep$, which is suitable by hypothesis. So $\gdep$ is in~$\wclo{\Sigma}$
  also in the second case. This concludes the proof.
\end{proof}

We can now conclude the proof of Theorem~\ref{thm:smallsaturation}, which will be direct from
Claim~\ref{clm:derived}:

\begin{proof}[Proof of Theorem~\ref{thm:smallsaturation}]
The running time bound was shown in
  Lemma~\ref{lem:computeclosure}, so we must only show that $\wclo{\Sigma}$ is a childish saturation as in
  Definition~\ref{def:satur}. We know from the easy direction of
  Claim~\ref{clm:derived} that all GTGDs in~$\wclo{\Sigma}$ are logically entailed
  by~$\Sigma$. In the rest of the proof, we show that $\wclo{\Sigma}$ is
  complete for fact entailment on childish instances.

  Let $I$ be a childish instance, let $F$ be a fact on the domain of~$I$
  such that $I, \Sigma \models F$.
  Let us show that $\inst, \wclo{\Sigma}
  \models F$. If $F \in \inst$, then there is nothing to show, so we assume that
  $F \notin \inst$. Now, considering a chase proof of $\inst,
  \Sigma \models F$, we know that $F$ must be isomorphic to the head of a GTGD
  of~$\Sigma$, namely, that of the last GTGD $\gdep_H$ which is
  fired in a chase proof of~$F$ from~$\inst$ with~$\Sigma$. Note that this uses homomorphism-closure of $\Sigma$, which implies that
  in rule firings, we can always assume that all exported variables of rules are mapped to
  distinct elements. We also know, from the width bound on~$\Sigma$, that $F$
  uses at most $w$ distinct elements of~$\inst$, and since $F$ is a fact on the
  domain of~$\inst$ it uses at most~$w$ distinct elements overall.

  The childish instance $\inst$ consists of a
  principal fact $F'$ which is an isomorphic copy of the head of
  some GTGD $\gdep_B$ from~$\Sigma$,
  together with some side facts $\inst'$ on at most $w$ elements of~$F'$.  Let
  $\gdep$ be the full GTGD obtained by renaming the elements of $\inst$ and~$F$
  from constants to variables, with $\inst$ giving the body and $F$ giving the
  head. We claim that $\gdep$ is a suitable derived full
  GTGD. Indeed:
  \begin{itemize}
    \item $\gdep$ has exactly one principal guard because $\inst$ contains
      exactly one principal fact $F'$ which contains all elements of
      $\adom(\inst)$.
    \item $\gdep$ is $\Sigma$-compatible, as witnessed by~$\gdep_H$ for the head
      atom and $\gdep_B$ for the principal guard atom.
    \item $\gdep$ has width at most~$w$, because $F$ uses at most $w$ distinct
      elements.
    \item $\gdep$ has breadth at most~$w$ because $I'$ uses at most $w$
      elements of~$F'$.
  \end{itemize}
  Thus, $\gdep$ is a suitable full GTGD. Further, $\gdep$ is a derived
  suitable full GTGD, because the chase proof of $I,
  \Sigma \models F$ witnesses that $\gdep$ is logically entailed by~$\Sigma$.
  Hence, by Claim~\ref{clm:derived}, we immediately conclude that $\gdep \in
  \wclo{\Sigma}$. Hence, $I, \wclo{\Sigma} \models F$, which is what we wanted
  to show.
\end{proof}

\section{Fact closure and making the chase similar to a linear chase} \label{sec:shortcut}
We are now ready for the second 
stage of our proof. For now, we have used an analysis
of the principal-exempt one-pass tree-like chase to design an algorithm that
computes a saturation of a finite set of GTGDs assuming it is given a childish instance.

We first show that we can use childish saturations to deal with the fact
entailment problem on arbitrary instances.

\begin{defi}
  Given an instance $\inst$ and GTGDs $\Sigma$, we say that $\inst$  is 
  \emph{$\Sigma$-fact-saturated}
if it is closed under facts entailed on the same domain.
  That is, $\inst$ is $\Sigma$-fact-saturated
  if for any fact $F$ over $\adom(\inst)$, if $\inst, \Sigma \models F$, then $F$ is already in $\inst$.

\end{defi}

\begin{defi} If  $\inst$ is a subinstance of $\inst'$ we say $\inst'$ is \emph{$\inst$-deactivated} if $\inst'$ has no active triggers with image in $\inst$.
\end{defi}

\begin{prop} \label{prop:factclosure} There is an algorithm that takes as input
  a set of GTGDs $\Sigma$ that strongly obeys side signature $\sign'$, along with
  an instance $\inst$, and performs chase steps with $\Sigma$ and with a
  childish saturation $\wclo{\Sigma}$ of~$\Sigma$ 
  to obtain an $\inst' \supseteq \inst$ that is $\Sigma$-fact-saturated and $\inst$-deactivated, 
  in
  time $\mathrm{Poly}(|\inst|^{O(w)}, |\Sigma|, a^{O(w)}, 2^{n' \times w^{a'}})$, where $a, w, n', a'$ are as in Lemma~\ref{lem:computeclosure}.
\end{prop}

\begin{proof}
  We can compute a childish saturation $\wclo{\Sigma}$
  in the required time, by Theorem~\ref{thm:smallsaturation}. We then perform a
  variant of the %
  one-pass chase
  of~$\Sigma \cup \wclo{\Sigma}$ over~$\inst$, but truncated to only one level,
  i.e., to the root node (initially a single node containing the facts
  of~$\inst$) and to the immediate children of the root
  node. More precisely, the process repeatedly applies each of the 
  following kinds of operations until saturation, starting with a singleton tree made of a root node
  $v_0$ containing the instance~$\inst$:
\begin{itemize}
  \item Firing all full rules of $\Sigma$ and all rules of $\wclo{\Sigma}$ on $v_0$  to create new side facts
    on~$\adom(\inst)$, which are added to~$v_0$ (we do this whenever such new facts can
    be added to~$v_0$).
  \item%
    Firing all non-full rules of~$\Sigma$ on every active trigger which is
    included in $\adom(\inst)$, creating child nodes with a principal fact which is the head of the rule. We do this whenever such an active trigger
    exists, and create a new child $v$ containing the head fact (a principal
    fact) and containing as inherited facts all side facts of~$v_0$ that are
    guarded by the principal fact. Further, we fire all full GTGDs of $\wclo{\Sigma}$ on~$v$. After this the trigger is no longer active.
  \item For triggers in $\inst$ of principal GTGDs we create  a new child for them, saturate using  $\wclo{\Sigma}$, and propagate new $\adom(\inst)$ facts back to the root~$v_0$.
  More precisely, we do the following whenever it can add
    a new fact to~$v_0$:
    \begin{itemize}
      \item Fire a trigger for a principal GTGD of~$\Sigma$~-- i.e., a rule with
        a principal atom in the head, full or non-full~--  to create a new child node $v$ of~$v_0$
        containing one principal fact $F$ and possibly some inherited
        side facts.
      \item 
        Fire all full GTGDs of~$\wclo{\Sigma}$ on~$v$ to create new
        facts on $\adom(F)$, which are added to~$v$.
    
        Note that we deviate from the principal-exempt chase here in that we may
        create new principal facts on~$\adom(F)$, in which case they are added
        to the current node, i.e., they do not trigger the creation of a child
        node with a relaxed chase step. 
  
      \item Propagate all facts in~$v$ from $\adom(F) \cap \adom(\inst)$,
        rootwards to~$v_0$. 
     
        Note that we deviate again from the principal-exempt chase here because
        we are allowed to propagate principal facts rootwards~-- but because we never inherit principal facts, they
        will remain in $v_0$ and will not be inherited.
    \end{itemize}
\end{itemize}
The final instance $\inst'$ that we take is the set of the facts in
  $v_0$ at the end of this process.

We first claim that this process finishes in the required time bound.
  In the first top-level bullet, the number
  of possible facts to add to $v_0$ is upper bounded by~$|\adom(\inst)|^{O(a')}$
  for~$a'$ the side signature arity, because we only create side facts in this bullet item. As $w \geq a'$, this is upper bounded by $|\adom(\inst)|^{O(w)}$.
  For the second top-level bullet, the number of nodes we create is bounded by the number of possible head instantiations of non-full GTGDs of~$\Sigma$, which can be bounded by choosing the non-full GTGD and the choice of exported elements, i.e., by $|\Sigma| \times |\adom(\inst)|^{O(w)}$. Further, on each created node we fire all full GTGDs of $\wclo{\Sigma}$ which can create at most $|\wclo{\Sigma}| \times a^{O(w)}$ new facts in each node.
  For the third top-level bullet, the number of possible facts to create in $\inst$
  is upper bounded by $|\wclo{\Sigma}|$ (the number of rule heads) times again
  $|\adom(\inst)|^{O(w)}$ (the number of instantiations of the exported
  variables of the firing that creates the fact), so the number of nodes that we create is upper bounded by
  $|\wclo{\Sigma}| \times |\adom(\inst)|^{O(w)}$. For each created node we fire all full GTGDs of $\wclo{\Sigma}$, for which the bound is the same as for the previous bullet item.
 
  So the number of chase steps performed satisfies the bound. Further, going
  over the GTGDs of $\Sigma$ and $\wclo{\Sigma}$ and testing possible
  applications is in polynomial time in their size and in the current number of facts created,
  because we can test applicability by mapping the principal guards of GTGDs to
  each choice of principal fact, and then checking if this defines a
  homomorphism from the GTGD body, and if so add the head if it is a new fact.
  Remembering that the size of $\wclo{\Sigma}$ is bounded, following the running
  time bound of Lemma~\ref{lem:computeclosure}, we conclude that the running
  time bound is as required.

  It is clear that the process only uses chase steps from $\Sigma \cup \wclo{\Sigma}$.
  
  We next argue that $\inst'$ is fact-saturated, which will rely on bullet items one and three, and we also will use the completeness
  of~$\wclo{\Sigma}$ on childish instances. Assume by contradiction that $\inst'$ is not fact-saturated.
  That is, there is a fact $F$ on $\adom(\inst')$ which
  is not in $\inst'$ but is entailed by~$\inst'$ and~$\Sigma$. There are two cases: either $F$ contains an element of $\adom(\inst') \setminus \adom(\inst)$, or it does not. In the first case, the image of the trigger in~$\inst'$ must be in one of the child nodes of~$v_0$ that we created; but these nodes contained a childish instance when they were created, and we have applied all GTGDs of $\wclo{\Sigma}$ to them, so because $\wclo{\Sigma}$ is complete on childish instances we must have also derived $F$ in~$\inst'$, which is a contradiction. Hence, let us focus on the second case, when $F$ is a fact over $\adom(\inst)$.
  
  In this case, consider a principal-exempt  one-pass chase proof of~$F$
  from~$\inst$. 
  We can assume without loss of
  generality that $F$ is a minimal counterexample,  in the sense that, in this proof, it is the
  first fact on~$\adom(\inst)$ which is derived in the proof but not present
  in~$\inst'$. The GTGD firing that created~$F$ cannot have been applied
  to the root node~$v_0'$ of the principal-exempt one-pass chase, otherwise all hypotheses to the firing are facts on~$\adom(\inst)$
  which were present in~$\inst'$ by minimality, and so they were present
  in~$v_0$, so that $F$ should have been
  derived by the first bullet item above. Hence, the GTGD firing that
  created~$F$ must have applied in a strict descendant of~$v_0'$ in the
  principal-exempt one-pass chase. This node, call it $b$, is a descendant of a child
  node~$b'$ of~$v_0'$, and $b'$ in turn was created by firing a principal GTGD $\gdep$
  of~$\Sigma$ on~$\inst$. The assumption that $\gdep$ is principal is because,
  by Proposition~\ref{prp:strongobey}, non-principal GTGDs are full and so
  firing them does not create new nodes in the principal-exempt chase.
  By minimality, all hypotheses to fire $\gdep$ were present
  in~$\inst'$ as well, so the same firing could have been performed in our
  truncated variant 
  of the chase.
  And by minimality we would have created a node inheriting the same facts
  from the root nodes, i.e., with the same set of facts $\inst''$ that
  were present in~$b'$ when it was created in the chase proof of~$F$ from~$\inst$.

  We now use the fact that $\inst''$ is a childish instance, because $\gdep$ is
  a principal GTGD of~$\Sigma$ of width at most~$w$. Thus, $\inst''$ consists of
  one principal fact which is an isomorphic copy of a GTGD head of~$\Sigma$,
  together with side facts on at most~$w$ elements, corresponding to
  the exported variables of $\gdep$. Thus, the completeness of~$\wclo{\Sigma}$ on the childish
  instance $\inst''$ ensures that the fact $F$, which is a fact
  on~$\adom(\inst'')$, was also derived by applying the rules of~$\wclo{\Sigma}$
  in our chase variant. This establishes a contradiction and concludes the proof of the second case of our case analysis. Hence, $\inst'$ is fact-saturated.

  We last argue that $\inst'$ is $\inst$-deactivated. Assume by way of contradiction that $\inst'$ contains an active trigger for a GTGD $\gdep$ with image in~$\inst$. If $\gdep$ is full, then this contradicts fact-saturation which we established above. If $\gdep$ is non-full, then this trigger should have been fired in the second bullet point, which is also a contradiction. Hence, $\inst'$ is $\inst$-deactivated, which concludes the proof.
\end{proof}

We note a corollary that may be of independent interest:

\begin{cor} \label{cor:factclosureptime} For any fixed side signature $\sign$ and width $w$, there is a polynomial time algorithm that takes as input
  a set of GTGDs $\Sigma$ that strongly obeys side signature $\sign'$
  and have width at most~$w$, along with
  an instance $\inst$, and computes a $\Sigma$-fact-saturated instance $\inst'
  \supseteq \inst$. %

  In particular, the  fact entailment  problem  is in polynomial time for such GTGDs.
\end{cor}
The corollary implies in particular that fact entailment for linear TGDs of bounded width is in polynomial time, a result that does not seem to have been noted explicitly in the literature.
Indeed, linear TGDs always vacuously strongly obey the empty side signature (up to
performing homomorphism-closure, which is in $2^w$ hence constant for
fixed~$w$).

If we were interested only in atomic queries -- single atoms with no quantifiers
--  instead of $\owqa$ with general CQs,
then Proposition~\ref{prop:factclosure} would suffice to conclude our main results. 
Note that for general CQ query
answering it is also possible to reduce the query answering process
to applying a set of full TGDs: general CQs are \emph{Datalog-rewritable} with respect to GTGDs, see~\cite{bbcjsl}.
But we do not see a way to use these methods to take advantage of the side signature and get the refined complexity bounds we need. 

Instead we will 
continue to use
tree-like chase proofs to do query answering, but look at the impact
of adding all of these derived GTGDs on the chase process.
Recall that in the tree-like chase we have chase steps and propagation steps.
Our goal is to get a
chase with no propagation steps at all,
as discussed in the proof overview of Section~\ref{sec:proofoverview}.
 To do that we will  add \emph{shortcuts} which summarize the impact of 
chase steps combined with propagation steps. We will do so by defining a
tree-like chase called the \emph{shortcut chase}. Note that this chase is
tree-like but it is \emph{not} a principal-exempt chase; it will not feature
relaxed chase steps, and will not create new nodes when firing full principal
GTGDs. The principal-exempt chase was only used in the previous section and in
Proposition~\ref{prop:factclosure}; we do not use it here and we will not use
it in the sequel.

Let $\Sigma$ be the set of GTGDs of width~$w$ that strongly obeys side signature
$\sidesign$, and let $\Sigma'$ be a childish saturation of~$\Sigma$. (We will of
course take $\Sigma'$ to be the closure $\wclo{\Sigma}$ defined in the previous
section.)
The \emph{shortcut chase} based on $\Sigma'$ is a tree-like chase sequence of
the form described below.
Informally we apply the set of
non-full GTGDs of $\Sigma$ and the full GTGDs of~$\Sigma'$ in alternation.

A shortcut chase starts with a chase tree consisting of a single root node $T_0$, then consists of two alternating kinds of steps:
\begin{itemize}
  \item The \emph{non-full steps}, where we fire a non-full GTGD of~$\Sigma$ on a
    node $g$. The chase step creates a new node~$g'$
    which is a child of~$g$, which contains the result $F'$ of firing the
    non-full GTGD along with a copy of the side facts of~$g$ which only use elements shared
    between $F$ and~$F'$ (i.e., all side facts that can be inherited are inherited).
  \item The \emph{full saturation steps}, which apply to a node~$g$, only once
    per node, precisely at the moment where it is created by a non-full step.
    In this step, we apply all the full GTGDs of~$\Sigma'$
    to the facts of~$g$, and add the consequences to~$g$ (they are still on the domain of~$g$
    because the rules are full).
\end{itemize}

Note the absence of propagation steps: the facts generated in a node are never propagated rootwards. Also note that whenever we create a node then the facts that it contains are a childish instance.

Our next goal is to argue that the 
shortcut chase is as good as the usual chase in terms of query answering. 
As we are working with childish saturations, which are only complete for entailment on
childish instances, we will use Proposition~\ref{prop:factclosure}
to ensure that the original instance is fact-saturated.

\begin{prop} \label{prop:shortcutchasecomplete} For any childish saturation
  $\Sigma'$ of $\Sigma$, the 
  shortcut chase based on $\Sigma'$
  emulates $\Sigma$
  on any fact-saturated instance: for each fact-saturated instance $\inst$ and
  each Boolean CQ $Q$, we have $\inst, \Sigma \models Q$ if and only if $Q$ holds in an instance produced from $\inst$ by a shortcut
chase sequence with the full GTGDs of~$\Sigma'$ and the non-full GTGDs
  of~$\Sigma$ starting at $\inst$.
\end{prop}
\begin{proof}
The direction from right to left is simple, and does not use fact-saturation
  of~$\inst$. Shortcut chase steps are chase steps, and the dependencies
  in the childish saturation~$\Sigma'$ are entailed by~$\Sigma$.
  Thus, it is clear that if $Q$ holds in an instance produced from $\inst$ using shortcut chase steps,
  then we have a chase proof witnessing that $\inst, \Sigma \models Q$.

We prove the other direction, assuming $\inst, \Sigma \models Q$.
  We first show the following auxiliary claim (*): in a 
  shortcut chase that starts with a fact-saturated instance, when we derive a fact $F$ in a
  node $v$ (which by fact-saturation is not the root node), then $F$ is not
  guarded in the parent of~$v$. Indeed, assume by contradiction that $F$ is
  guarded in a node $v'$ of~$F$ but was derived in a child $v$ of~$v'$.
  The set of facts present in the node~$v'$ when it was created formed a childish
  instance $I'$, and the 
  shortcut chase sequence witnesses
  that the fact $F$ on $\adom(I')$ is entailed by~$I'$ and~$\Sigma'$, hence
  by~$I'$ and~$\Sigma$. So by completeness of the childish saturation $\Sigma'$,
  we should have derived $F$ in the full saturation step on~$v$ and
  inherited it in~$v'$ instead of deriving it on~$v'$ as we assumed --- a contradiction.

  We next show that (*) implies the following claim (**): in a
  shortcut chase that starts with a fact-saturated
  instance, every side fact
  occurs in every node in which it is guarded. To show this, 
  first notice that side facts which are in the initial instance will
  be inherited as long as they are guarded, so there is nothing to show for
  them. As for side facts which are not in the initial instance,
  they are always created by full saturation steps, because the
  GTGDs of~$\Sigma$ strongly obey the side signature, so non-full GTGDs are all principal.
  Let us consider a side fact $F$ derived on a node $v$ by a
  full saturation step. When this happens, $v$ is a leaf node, and we
  claim that there is no other
  node $v'$ that guards $F$.
  Indeed, if there were another node $v'$ which does guard the fact, then
  let $v''$ be the least common ancestor of $v$ and $v'$. It must be the case
  that $v''$ also guards $F$, otherwise the values of $\adom(F)$ cannot have
  been re-introduced in~$v$ and in~$v''$ independently. So since $v$ is a leaf
  node when $F$ is created, we have $v \neq v''$. This implies that, considering
  the leafward path from~$v''$ to~$v$, the parent of~$v$ also guards $\adom(F)$.
  But it does not contain~$F$, and we have a contradiction of (*). Thus, when
  $F$ is created in occurs in the only node in which it is currently guarded.
  Then, as $F$ is a side fact, it will be inherited leafwards in all nodes where
  it is guarded. 
  Thus, we have established (**).

We are now ready to show that the 
shortcut chase is complete,
  using (**). We show that for every tree-like chase proof $T_0, \ldots, T_m$
  which starts from
  a fact-saturated instance $\inst$, there is a %
  shortcut chase
  proof $T_0', \ldots, T_m'$ from $\inst$  with an isomorphism from the
  underlying instance of $T_m$ to that of $T_m'$
  which is the identity on the active domain of $\inst$. Thus, in particular, if $T_0, \ldots, T_m$
  witnesses that $\inst, \Sigma \models Q$, then 
  $Q$ holds in the instance underlying $T_m$, and taking the image of a match
  by~$\iota_m$ we see that the same is true on $T_m'$.

  Let us proceed by induction on the length of the tree-like chase
  proof $T_0, \ldots, T_m$. The base case is a tree-like chase proof of
  length~$0$, in which case we simply pick $\iota_0$ to be the identity on~$\inst$.

  For the induction step, consider a tree-like chase
  proof $T_0, \ldots, T_m$ with $m>0$. We immediately apply the induction hypothesis to obtain a
   shortcut chase proof $T_0', \ldots, T_{m'}'$ 
  such that there is an isomorphism
  $\iota_{m-1}$ from the underlying instance of~$T_{m-1}$
  to that of~$T_{m'}$ which is the identity on~$\adom(\inst)$, and let us explain how to continue the 
  shortcut chase proof to obtain a new tree $T_{m'+1}'$ and an isomorphism
  $\iota_m$ from the underlying instance of~$T_m$ to that of~$T_{m'+1}'$ which
  is the identity on $\adom(\inst)$.
  
  We first eliminate two easy cases. First, $T_m$ may be obtained by applying a
  propagation step. Second, $T_m$ may be obtained by performing a chase step in
  a node $v$ which derives a fact that already occurs somewhere in $T_{m-1}$. In
  both these cases,
  we can conclude immediately 
  using $T_0', \ldots, T_{m'}'$ and $\iota_{m-1}$ as a witness. So we assume that
  $T_m$ is obtained by applying a chase step, and let $F$ be the fact produced. We distinguish
  two cases depending on whether the GTGD $\gdep$ used in the chase step is full or
  non-full.

  First, if the GTGD $\gdep$ is full, consider the image of its trigger in
  $T_{m-1}$,
  calling it $S$. Because $\gdep$ is a rule of $\Sigma$ which strongly obeys
  the side signature, we know that $S$ is
  formed of a principal fact $F_0$ and side facts $\beta_0$ on $\dom(F_0)$. 
  Given that $F$ is a new fact, as the instance $\inst$ is
  fact-saturated, we know that $\adom(F) \not\subseteq \adom(\inst)$, so since
  $\gdep$ is full we know that
  $\adom(F_0) \not\subseteq \adom(\inst)$. So
  letting $F_0'
  := \iota_{m-1}(F_0)$, we know that $F_0'$ occurs in a node $n'$ of~$T_{m'}'$ which is not
  the root. The facts of $\iota(\beta_0)$ also occur in $T_{m'}'$ and are
  guarded by~$n'$, so by (**) they also occur in~$n'$. If $n'$ already
  contains  the fact $\iota_{m-1}(F)$ within $T_{m'}'$, then we can just conclude
  immediately with
  $T_0', \ldots, T_{m'}'$ and $\iota_{m-1}$. Otherwise, 
  we can continue the 
  shortcut chase proof by 
  firing~$\gdep$ on~$n'$ 
  which contains the trigger $\iota_{m-1}(S)$, and deduce the fact
  $\iota_{m-1}(F)$ in~$n'$.
  This gives us $T_{m'+1}'$ which admits an isomorphism $\iota_m :=
  \iota_{m-1}$ from $T_{m'}$ to $T_{m'+1}'$ which is the identity on~$\adom(\inst)$.

We now consider the case where the GTGD $\gdep$ is non-full. Consider again the image $S$ of its trigger
  in~$T_{m-1}$.
  We want to find a node $n'$ of~$T_{m'}'$ on which the same firing can be applied.
  Again $S$ is formed of a principal fact $F_0$ and side facts
  $\beta_0$ on $\dom(F_0)$. This time it may be the case that $F_0$ is a fact
  of~$\inst$. In this subcase $\beta_0$ are facts over $\dom(I)$ which are entailed
  by~$\inst$ so $\beta_0 \subseteq \inst$ because $\inst$ is fact-saturated. So in this subcase
  we can take $n'$ to be the root node of $T_{m'}'$, which contains the facts of
  $\iota_{m-1}(S)$. In the subcase where $F_0$ is not a fact of~$\inst$, then we reason as
  in the previous case: $\iota_{m-1}(F_0)$ occurs in a non-root node 
  of~$T_{m'}'$, the facts of $\iota_{m-1}(\beta_0)$ are guarded in that node, so also
  occur there by (**). In this subcase we define $n'$ to be a node of $T_{m'}'$ that contains the
  facts of~$\iota_{m-1}(S)$. 
  
  In both these subcases we can perform a non-full step
  and apply $\gdep$ with trigger mapping to $\iota_{m-1}(S)$, within node~$n'$
  to create the tree $T_{m'+1}'$ with a child $n''$ of~$n'$ and a fact $F'$
  isomorphic to~$F$. We can extend
  $\iota_{m-1}$ to $\iota_m$ by mapping the new values introduced in the chase
  step from $T_{m-1}$ to~$T_m$ to the isomorphic values in~$F'$.  This gives
  $\iota_m$, which is an isomorphism from the underlying instance of~$T_m$ to
  that of~$T_{m'+1}'$, and is the identity on~$\adom(\inst)$ as required.

  So, both for full GTGDs and non-full GTGDs, we have successfully extended the
  tree-like chase
  sequence while preserving an isomorphism which is the identity
  on~$\inst$, so every query that admits a tree-like proof also admits a proof with
  the 
  shortcut chase. This concludes the proof.
\end{proof}

We will explain in the next section how to translate the shortcut chase to a set of linear TGDs.

\section{The final stage: Linearization and its justification} \label{sec:linearize}

We now describe the third and last stage of the proof of
Theorem~\ref{thm:idreduce}. We describe the translation to a set of linear TGDs,
intuitively introducing predicates for each guarded set of atoms on the side
signature; and also describe the pre-processing of the instance, intuitively
closing under the full dependencies that we introduce.
The correctness of this transformation will then be argued using the shortcut
chase studied in the previous section.

Recall the definition of childish instances (Definition~\ref{def:childish}).
In the definition of the linearization, we will refer to \emph{types}:
\begin{defi}
  A \emph{type} is an isomorphism type of a childish instance, i.e., an
  equivalence class of the childish instances quotiented under the isomorphism
  equivalence relation.
\end{defi}

For each type $\theta$, create a predicate $R_\theta$ whose arity is that of
the principal fact of~$\theta$.
Observe that this
creates a singly exponential number of relations when
the arity~$a'$ of~$\sidesign$ is fixed, and it creates 
only polynomially many relations when we further fix the  arity of the signature and also fix  the full side
signature
$\sidesign$.
We then let $\linearize(\Sigma)$ consist of the linear TGDs defined as follows.
For every type $\theta$, fix a childish instance $S$ on domain $\vec x$ achieving
type~$\theta$, and write $P(\vec x)$ for the principal fact of~$S$ and 
$\vec y \subseteq \vec x$ for the elements of $\vec x$
on which there are
side facts (there are at most~$w$).
Also fix a homomorphism\footnote{Note that we could avoid considering such
homomorphisms if we required that the saturation is homomorphism-closed: this
could be achieved in the required bounds just as in the case of the input
constraints. However, in the present section we think it is simpler to
consider the homomorphisms directly.} $h$ from~$\vec y$ to itself, and
extend $h$ to  a homomorphism from~$\vec x$ to itself which is the
identity on elements that do not occur in~$\vec y$.
Let $S'$ be the instance obtained by starting with $h(S)$ and then repeatedly
applying the full TGDs
of~$\wclo{\Sigma}$ (a full saturation step).
Then our linearization will contain:

\begin{itemize}
  \item (Instantiate): 
    For every fact $A(h(\vec x))$
    of~$S'$,
    the full linear TGD:
    \[
      R_{\theta}(h(\vec x)) \rightarrow A(h(\vec x))
    \]
    Note that in particular we always have the TGD:
    \[
      R_{\theta}(h(\vec x)) \rightarrow P(h(\vec x))
    \]
  \item (Lift): 
    For every non-full GTGD $\dep$ in $\Sigma$ of the form $\body(\vec
    z) \rightarrow \exists \vec w ~ T(\vec z, \vec w)$ where
    $\body(\vec z)$ 
    has a match $h'$ in~$S'$, we
    add the linear TGD:
    \[
      R_{\theta}(h(\vec x)) \rightarrow \exists \vec w ~ R_{\theta''}(h'(\vec z), \vec w)
    \]
    where $\theta''$ is the type of the childish instance obtained as follows:
    \begin{itemize}
        \item start with a principal fact $F = T(h'(\vec z), \vec a)$ where $\vec a$ are nulls, intuitively corresponding to the result of firing~$\dep$;
        \item add the side facts of $S'$ that use only elements of~$h'(\vec z)$, intuitively corresponding to inherited facts.
    \end{itemize}
\end{itemize}

Informally, the (Lift) rules simulate firing of non-full rules in the shortcut chase, and
the (Instantiate) rules move us back into the original signature from the lifted signature.

 Note that the  full GTGDs are not explicitly mentioned in these linearized rules. However, they play a role in two places. First, within (Instantiate) and (Lift) rules, we are computing derived facts in computing $S'$, and this involves the full GTGDs. Secondly, we will be running these rules on a   pre-processed instance which is already fact-saturated.

We give a brief example to show how the transformation works:
\begin{exa} \label{ex:linearize}

Let $\Sigma$ consist of the ID
$R(x,y) \rightarrow \exists z ~ R(y,z)$
and the full GTGD
$R(x,y), U(x) \rightarrow U(y)$.

\begin{itemize}
\item Our side signature here will consist of only $U$; thus the only principal
  relation is $R$;
\item The maximal arity of a side signature atom, denoted $a'$, is $1$;
\item The maximal width $w$ is $1$, which is $\geq a'$.
\end{itemize}

We have three childish instances up to isomorphism: 
\begin{itemize}
    \item $\inst_1=\{R(1,2), U(1)\}$, corresponding to type $\theta_1$
    \item $\inst_2=\{R(1,2), U(2)\}$, corresponding to type $\theta_2$
    \item $\inst_3=\{R(1,2)\}$, corresponding to type $\theta_3$
    \end{itemize}

    In the linearization we thus have fresh relations $R_{\theta_1}$ and
    $R_{\theta_2}$ and $R_{\theta_3}$.

The rule (Instantiate) will give us full linear TGDs such as:
\begin{align*}
 R_{\theta_1}(x,y) \rightarrow U(x) \\
 R_{\theta_1}(x,y) \rightarrow R(x,y)
\end{align*}

The rule (Lift)  will provide us with linear TGDs such as:
\[
R_{\theta_2}(x,y) \rightarrow \exists z ~ R_{\theta_1}(y,z).
\]
We will show in the sequel that the resulting linear TGDs allow us to solve
  $\owqa$ for the original dependencies, intuitively because they amount to
  performing the shortcut chase.
\end{exa}

The result of our transformation clearly consists of linear TGDs. Further, they are
of semi-width~$w$: indeed, the rules produced by (Lift) have
width bounded by~$w$ because the same is true of the non-full GTGDs
of~$\Sigma$, and the rules produced by (Instantiate) have an acyclic position graph.

We have now given the algorithm for computing $\linearize(\Sigma)$.
We now argue that
these rules can be computed efficiently:

\begin{clm}
  \label{clm:linearizebound}
  For any constant bound $a' \in \mathbb{N}$ on the arity of the side signature,
  there are fixed polynomials $P_1$ and $P_2$ such that the number of
rules 
in $\linearize(\Sigma)$
  is in $P_1(\card{\Sigma} \times
a \times w)^{P_2(w, n')}$ and the time to construct them is polynomial in their
number.
\end{clm}

The proof for this is similar to the argument for the bound
shown in Lemma~\ref{lem:numbersuitable}.

\begin{proof}
Let $n$ be the
number of relations of the signature, $a$ the maximal arity of the signature,
$n'$ the number of relations in the side signature,  and $a'$ the maximal arity
of the side signature.
We can bound the number of relations $R_\theta$ by bounding the number of types, which amounts to bounding the number of childish instances up to isomorphism. Now, a childish instance is defined by choosing the GTGD head of~$\Sigma$ that we copy (a factor of $|\Sigma|$), the
subset of variables on which to add side facts (a factor of
$a^{O(w)}$), and a set of side facts on those elements (a factor of at
most $2^{n' \cdot w^{a'}}$).

Let us now bound the number of rules.
The number of rules obtained by (Instantiate) is bounded by the number of relations
$R_\theta$ (bounded above), times the number of choices for the homomorphism $h$ (i.e., $w^w$), times
the number of possible choices for the fact in the rule head. The number of such choices can be bounded by observing
that such facts are obtained by instantiating the head of a GTGD
of~$\wclo{\Sigma}$ (so at most $|\wclo{\Sigma}|$ choices), identifying some elements in the head (of which there are at most~$w$, so a factor $w^w$ to choose a
homomorphism), and selecting the tuple of elements on which the fact is created (i.e., $a^w$).

  Further, the number of rules created by (Lift) is bounded by the number of relations $R_\theta$ (bounded above), times the number of homomorphisms (i.e., $w^w$), times $a^{O(w)}$ for the choice of exported variables, times the number of
relations $R_{\theta''}$ for the choice of head.

Note that in the bounds calculated above, if we treat $a'$ as a constant, every exponent is either $O(w)$ or $O(n', w)$. The number of rules is a polynomial in these quantities. Hence, having fixed the arity $a'$ of the side signature,
there are indeed fixed polynomials $P_1$ and $P_2$ satisfying our claim.
\end{proof}

We have described the construction of the new linear constraints $\linearize(\Sigma)$.

For any instance $\inst$ in the original signature,  we construct
$\inst^\lift$ by considering 
every childish instance $S$ on domain $\vec a$ which is a subinstance of $\inst$ and adding
the fact $R_\theta(\vec a)$
where $\theta$ is the type of~$S$.
This still respects the time
bounds, because it is doable in $\ptime$ in~$|\inst|^w \times |\Sigma|$.

The last thing to show is that the linearization is correct:
\begin{prop} \label{prop:linearizecorrect} Let $\inst$ be an instance and $\inst'$ be a superinstance formed via $\Sigma \cup \wclo{\Sigma}$ chase steps that is fact-saturated
  and $\inst$-deactivated for~$\Sigma$. Then
$\inst', \Sigma \models Q$ if and only if $(\inst')^\lift, \linearize(\Sigma) \models Q$.
\end{prop}

\begin{proof}
To prove the result, for an instance $\inst''$ of the signature introduced above, let $\delinearize(\inst'')$ be the result
of replacing facts over predicates $R_\theta$ by the facts
  corresponding to $\theta$. Note that, using the rules created by (Instantiate),
every $R_\theta$-atom entails the corresponding $\theta$-atoms.

The ``soundness direction'', from right to left, is easy to see, and it does not use fact-saturatedness of the instance $\inst'$ or the fact that $\inst'$ is $\inst$-deactivated. 
Indeed, for any linearized rule
$\sigma$, suppose a chase step with $\sigma$ on $(\inst')^\lift$ yields $\inst''$.
Then there are ordinary chase steps using rules in $\Sigma$  that produce $\delinearize(\inst'')$ from
$\inst'$.

To prove the ``completeness direction'', we show that \emph{for any shortcut chase sequence with $\Sigma$ on $\inst'$,  there is a chase sequence  with $\linearize(\Sigma)$ whose delinearization will produce all the same facts}.
Since the shortcut chase
captures entailment on fact-saturated instances (from Proposition~\ref{prop:shortcutchasecomplete}),
this is sufficient to conclude.

To see this, consider the shortcut chase where each non-root node $v$ is annotated by the type of the childish instance that this node had when it was created from its parent.
 Consider also the chase by the rules obtained with (Lift) in~$\linearize(\Sigma)$. We can observe
by a straightforward induction that there is a chase by the rules of (Lift) in~$\linearize(\Sigma)$ such that the tree structure on the nodes of the shortcut chase
with the indicated annotations is isomorphic to the tree of facts of the form
  $R_{\theta}(\vec x)$ created in this lifted chase. 
  
  Here, we use the fact that $\inst'$ is $\inst$-deactivated in the following way. Consider the first round of non-full shortcut chase steps applied on all active triggers in~$\inst'$. Consider specifically the firing of a non-full GTGD $\gdep$ on a trigger $\trig$ which creates a node $b$ in the first round of the shortcut chase. Our goal is to find a childish subinstance $S$ of~$\inst'$ that we can use to replicate this firing in the linearization.
  
  We know that the image $\iota$ of the trigger $\trig$ cannot be in~$\inst$ because $\inst'$ is $\inst$-deactivated. So $\iota$ must involve at least one fact of $\inst'\setminus \inst$. In fact, as $\inst'$ is fact-saturated, the principal fact $F$ that guards $\iota$ must be a fact of~$\inst'\setminus \inst$.
  But as $\inst'$ was created from~$\inst$ by chase steps, we know that $F$ was created in $\inst'$ by firing a GTGD $\gdep'$ of $\Sigma \cup \wclo{\Sigma}$. If the GTGD $\gdep'$ is full, then the width bound ensures that $F$ contains at most $w$ distinct elements, so we can take $S \coloneq \iota$ and $S$ is a childish subinstance of~$\inst'$. If the GTGD $\gdep'$ is non-full, then $F$ was created in a new chase node together with inherited side signature facts on at most $w$ elements, so $F$ was part of a childish subinstance $S$ of~$\inst'$ at that point; and $S$ entails all facts guarded by~$F$ which are derived afterwards by chase steps in~$\inst'$, in particular those of~$\iota$. 
  
  Letting $S$ be the childish instance formed above, let $\theta$ be the type of~$S$.
  By construction of $(\inst')^\lift$, in both cases there is a fact $F_S$ for the relation $R_\theta$ on the elements of the childish instance $S$, and the firing of the active trigger $\trig$ in the first round of the shortcut chase can be performed in the linearization by firing a rule on $F_S$ whose body uses relation $R_\theta$ and whose head uses the relation $R_{\theta'}$ for $\theta'$ the type of the childish instance created when firing $\gdep$ on~$\trig$. This allows us to label the node $b$ of the shortcut chase with~$\theta'$.

  For subsequent rounds of the shortcut chase steps, we know that active triggers are contained in nodes that were created by a non-full chase step in the previous round and were then fully saturated; and when these nodes are created they contain a childish instance and can be labeled accordingly.

  Now, the
applications of the rules 
  obtained with (Instantiate) to the linearized chase create precisely the facts in the original signature
contained in these nodes, so this shows that the chase by~$\linearize(\Sigma)$ and the
shortcut chase create  the same facts.

This shows that~$\linearize(\Sigma)$ satisfies indeed the hypotheses of
Theorem~\ref{thm:idreduce}.
\end{proof}

  We finish by putting together the steps of the proof:

\begin{proof}[Proof of Theorem~\ref{thm:idreduce}.]
  We have already explained at the end of Section~\ref{sec:proofoverview} how we
  reduce to the setting where $\Sigma$ strongly obeys a side signature
  $\sidesign$ of suitable size.
  
  We use Theorem~\ref{thm:smallsaturation} to compute a childish saturation
  $\wclo{\Sigma}$ of~$\Sigma$. The running time bound is now polynomial
  in~$|I_0|$, $(|\Sigma| \times a)^{O(w)}$, and $2^{n' \times w^{a'}}$.
  Using Proposition~\ref{prop:factclosure}, we perform chase steps to construct $I_0'$ which is 
  $\Sigma$-fact-saturated  and $I_0$-deactivated, 
  with the complexity being now polynomial in the previous values and in
  $|I_0|^{O(w)}$. Last, we compute the linearization $(I_0')^\lift$ and
  $\Sigma^\lift$ as explained in this
  section. By Claim~\ref{clm:linearizebound}, the overall complexity is as claimed.
  Further, as we explained, the obtained dependencies are linear TGDs of
  semi-width $\leq w$ and arity $\leq a$.
  Last, we know by Proposition~\ref{prop:linearizecorrect} that the
  linearization is correct, i.e., $(I_0')^\lift$ and $\Sigma^\lift$ emulate $I_0'$
  and~$\Sigma$, concluding the proof.
\end{proof}

\section{Conclusion} \label{sec:conc}

In this work we have given finer bounds on query answering with GTGDs, based on the notion of side signature.
In addition to justifying claims in prior papers, we believe it is a good
example of the use of a linearization technique similar to the one
from~\cite{gmp}, and also a good example of how to use the recently-developed notion of the one-pass chase.
We hope that this combination could be used to get finer-grained bounds for query answering for other
TGD classes, such as frontier-guarded TGDs~\cite{baget2010walking}.

While our approach produces linearizations in the same style as \cite{gmp}, our techniques seem to be quite different -- e.g., we rely heavily on the completeness of specialized versions of the chase. It would be quite interesting to approach these bounds using the machinery and terminology of \cite{gmp}, but we leave this for future work.

When we do not impose any width bound, but merely fix the side signature arity,
we show an $\exptime$ bound.
Not only is the problem for this class $\exptime$-complete, but there is an $\exptime$-hardness result in
 \cite{bbbicdt} for the case with the simplest possible fixed side signature: one unary predicate.
In that sense, our $\exptime$ result is optimal.
We do not know whether our two results (Result~\ref{res:exptime}
and~\ref{res:np})
could be extended to the setting of multi-head TGDs or TGDs with constants in
full generality (see discussion in Section~\ref{sec:results} and see Appendix~\ref{apx:constantmulti}).

Our results give new classes where $\owqa$ is in $\exptime$ and new cases where it is in $\np$. 
We do not provide larger classes where the complexity is in $\pspace$: this is a bit surprising, and is 
one area of particular interest for future work.
Indeed, the canonical example where $\owqa$ is in $\pspace$ is when TGDs are linear, and our technique works by 
reducing to that case. Hence, it is natural to ask whether our techniques can
give a $\pspace$ upper bound (rather than $\exptime$) for some more general
classes defined via the notion of side signature.

\bibliographystyle{alphaurl}
\bibliography{algs}

\newcommand{\etalchar}[1]{$^{#1}$}
\begin{thebibliography}{FKMP05}

\bibitem[AB18]{resultlimitedpods}
Antoine Amarilli and Michael Benedikt.
\newblock When can we answer queries using result-bounded data interfaces?
\newblock {\em PODS}, 2018.
\newblock \href {https://doi.org/10.1145/3196959.3196965}
  {\path{doi:10.1145/3196959.3196965}}.

\bibitem[AB22]{resultlimitedj}
Antoine Amarilli and Michael Benedikt.
\newblock When can we answer queries using result-bounded data interfaces?
\newblock {\em LMCS}, 18(2), 2022.
\newblock \href {https://doi.org/10.46298/LMCS-18(2:14)2022}
  {\path{doi:10.46298/LMCS-18(2:14)2022}}.

\bibitem[BBB13]{bbbicdt}
Vince B\'ar\'any, Michael Benedikt, and Pierre Bourhis.
\newblock {Access patterns and integrity constraints revisited}.
\newblock In {\em ICDT}, 2013.
\newblock \href {https://doi.org/10.1145/2448496.2448522}
  {\path{doi:10.1145/2448496.2448522}}.

\bibitem[BBG{\etalchar{+}}22]{gsatvldb}
Michael Benedikt, Maxime Buron, Stefano Germano, Kevin Kappelmann, and Boris
  Motik.
\newblock Rewriting the infinite chase.
\newblock {\em PVLDB}, 2022.
\newblock \href {https://doi.org/10.14778/3551793.3551851}
  {\path{doi:10.14778/3551793.3551851}}.

\bibitem[BBG{\etalchar{+}}24]{gsatvldbjournal}
Michael Benedikt, Maxime Buron, Stefano Germano, Kevin Kappelmann, and Boris
  Motik.
\newblock Rewriting the infinite chase for guarded {TGD}s.
\newblock {\em TODS}, 2024.
\newblock \href {https://doi.org/10.1145/3696416} {\path{doi:10.1145/3696416}}.

\bibitem[BBJT19]{privacyijcai}
Michael Benedikt, Pierre Bourhis, Louis Jachiet, and Micha{\"{e}}l Thomazo.
\newblock Reasoning about disclosure in data integration in the presence of
  source constraints.
\newblock In {\em IJCAI}, 2019.
\newblock \href {https://doi.org/10.24963/IJCAI.2019/215}
  {\path{doi:10.24963/IJCAI.2019/215}}.

\bibitem[BBtC18]{bbcjsl}
Vince B{\'a}r{\'a}ny, Michael Benedikt, and Balder ten Cate.
\newblock Some model theory of guarded negation.
\newblock {\em Journal of Symbolic Logic}, 2018.
\newblock \href {https://doi.org/10.1017/JSL.2018.64}
  {\path{doi:10.1017/JSL.2018.64}}.

\bibitem[BLM10]{baget2010walking}
Jean{-}Fran{\c{c}}ois Baget, Michel Lecl{\`{e}}re, and Marie{-}Laure Mugnier.
\newblock Walking the decidability line for rules with existential variables.
\newblock In {\em KR}, 2010.

\bibitem[{Bor}06]{motikthesis}
{Boris Motik}.
\newblock {\em Reasoning in description logics using resolution and deductive
  databases}.
\newblock PhD thesis, { Karlsruhe Institute of Technology}, 2006.

\bibitem[CF05]{guardedconstants}
Balder~ten Cate and Massimo Franceschet.
\newblock Guarded fragments with constants.
\newblock {\em Journal of Logic, Language and Information}, 14(3), 2005.
\newblock \href {https://doi.org/10.1007/s10849-005-5787-x}
  {\path{doi:10.1007/s10849-005-5787-x}}.

\bibitem[CGK13]{tamingjournal}
Andrea Cal\`{\i}, Georg Gottlob, and Michael Kifer.
\newblock {Taming the infinite chase: Query answering under expressive
  relational constraints}.
\newblock {\em JAIR}, 2013.
\newblock \href {https://doi.org/10.1613/JAIR.3873}
  {\path{doi:10.1613/JAIR.3873}}.

\bibitem[CGL12]{datalogpmj}
Andrea Cal\`i, Georg Gottlob, and Thomas Lukasiewicz.
\newblock {A general {D}atalog-based framework for tractable query answering
  over ontologies}.
\newblock {\em Journal of Web Semantics}, 14, 2012.
\newblock \href {https://doi.org/10.1016/J.WEBSEM.2012.03.001}
  {\path{doi:10.1016/J.WEBSEM.2012.03.001}}.

\bibitem[CGLP11]{datalogpm}
Andrea Cal\`{\i}, Georg Gottlob, Thomas Lukasiewicz, and Andreas Pieris.
\newblock A logical toolbox for ontological reasoning.
\newblock {\em SIGMOD Record}, 40(3), 2011.
\newblock \href {https://doi.org/10.1145/2070736.2070738}
  {\path{doi:10.1145/2070736.2070738}}.

\bibitem[CLR03]{calirewriting}
Andrea Cal\`i, Domenico Lembo, and Riccardo Rosati.
\newblock Query rewriting and answering under constraints in data integration
  systems.
\newblock In {\em IJCAI}, 2003.

\bibitem[DLN07]{dln}
Alin Deutsch, Bertram Lud\"ascher, and Alan Nash.
\newblock {Rewriting queries using views with access patterns under integrity
  constraints}.
\newblock {\em TCS}, 371(3), 2007.
\newblock \href {https://doi.org/10.1016/J.TCS.2006.11.008}
  {\path{doi:10.1016/J.TCS.2006.11.008}}.

\bibitem[FKMP05]{fagindataex}
Ronald Fagin, Phokion~G. Kolaitis, Renee~J. Miller, and Lucian Popa.
\newblock {Data exchange: {S}emantics and query answering}.
\newblock {\em TCS}, 336(1), 2005.
\newblock \href {https://doi.org/10.1016/J.TCS.2004.10.033}
  {\path{doi:10.1016/J.TCS.2004.10.033}}.

\bibitem[GMP14]{gmp}
Georg Gottlob, Marco Manna, and Andreas Pieris.
\newblock Polynomial combined rewritings for existential rules.
\newblock In {\em KR}, 2014.

\bibitem[GMP20]{gottlob2020multi}
Georg Gottlob, Marco Manna, and Andreas Pieris.
\newblock Multi-head guarded existential rules over fixed signatures.
\newblock In {\em {KR}}, 2020.
\newblock \href {https://doi.org/10.24963/KR.2020/45}
  {\path{doi:10.24963/KR.2020/45}}.

\bibitem[JK84]{johnsonklug}
David~S. Johnson and Anthony~C. Klug.
\newblock {Testing containment of conjunctive queries under functional and
  inclusion dependencies}.
\newblock {\em JCSS}, 28(1), 1984.
\newblock \href {https://doi.org/10.1016/0022-0000(84)90081-3}
  {\path{doi:10.1016/0022-0000(84)90081-3}}.

\bibitem[Kap19]{kevinarxiv}
Kevin Kappelmann.
\newblock Decision procedures for guarded logics, 2019.
\newblock \url{https://arxiv.org/abs/1911.03679}.

\bibitem[Lib95]{libkin1995elements}
Leonid Libkin.
\newblock {\em Elements of finite model theory}.
\newblock Springer, 1995.

\bibitem[LMPS15]{lukasiewicz2015classical}
Thomas Lukasiewicz, Maria~Vanina Martinez, Andreas Pieris, and Gerardo~I.
  Simari.
\newblock From classical to consistent query answering under existential rules.
\newblock In {\em {AAAI}}, 2015.
\newblock \href {https://doi.org/10.1609/AAAI.V29I1.9414}
  {\path{doi:10.1609/AAAI.V29I1.9414}}.

\bibitem[MMS79]{maier}
David Maier, Alberto~O. Mendelzon, and Yehoshua Sagiv.
\newblock Testing implications of data dependencies.
\newblock {\em TODS}, 4(4), 1979.
\newblock \href {https://doi.org/10.1145/320107.320115}
  {\path{doi:10.1145/320107.320115}}.

\bibitem[One13]{onet}
Adrian Onet.
\newblock The chase procedure and its applications in data exchange.
\newblock In {\em Data Exchange, Integration, and Streams}, 2013.
\newblock \href {https://doi.org/10.4230/DFU.VOL5.10452.1}
  {\path{doi:10.4230/DFU.VOL5.10452.1}}.

\end{thebibliography}
\appendix
\section{Proof of the semi-width result
(Proposition~\ref{prop:semiwidthclassic-general})}
\label{apx:semiwidthproof}

In this appendix, we prove the $\np$ bound on $\owqa$ for bounded semi-width
linear TGDs, i.e., 
Proposition~\ref{prop:semiwidthclassic-general}. Recall its statement:

\semiwidthclassic*

The proof presented here is the same as
in~\cite[Appendix~C]{resultlimitedj}, except that to go from IDs to linear
TGDs we must change the statement of proof of Lemma~\ref{lem:depthbound}. The
proof is otherwise identical up to minor changes.

To prove the result, let $\Sigma$ be the collection of linear TGDs. We will
reason about the tree-like chase sequences (see Section~\ref{sec:onepass}) that can be obtained starting with some
instance~$I_0$. 
Specifically, in this appendix, when talking about a \emph{tree-like chase
sequence}, we mean the following:
we always consider \emph{relaxed} tree-like
chase sequences, where chase steps performed with full TGDs are \emph{always}
relaxed (i.e., the fact is always created in a new child node);
and further we \emph{never perform propagation steps} (they are never needed)
and we \emph{never inherit any facts} when creating fresh nodes.
This is consistent with how tree-like chase sequences work in the case of IDs considered in Johnson and Klug's work \cite{johnsonklug}.

Thanks to this assumption, in the tree-like chase sequences that we consider, the root
node always contains precisely the facts of~$I_0$, and the non-root nodes contain
precisely one fact and are in one-to-one correspondence with the facts that are
generated. Further, if we fired a trigger whose image is a fact~$F$, and
this creates a fact $F'$,
then the node $n'$ created by applying the chase step so that $T(n') = \{F'\}$
is a child of the node $n$ such that $T(n) = \{F\}$.

Let us now consider a chase tree~$T$ within some tree-like chase proof starting
with~$I_0$. A \emph{generated fact} in $T$
is a fact which is not a fact of~$I_0$. 
Let us now consider nodes $n$ and $n'$ in~$T$,  with~$n$ a strict ancestor of~$n'$. 
We say $n$ and $n'$ are \emph{far apart}
 if there are distinct generated facts $F_1$ and $F_2$ such that:
\begin{itemize}
\item  the node $n_1$
 corresponding to~$F_1$ and the node~$n_2$ corresponding to~$F_2$ are both ancestors
 of~$n'$ and descendants of~$n$;
\item  $n_1$ is an ancestor of~$n_2$;
\item  $F_1$ and $F_2$ were generated by the same rule of~$\Sigma$; and
\item the equalities between values in positions within~$F_1$ are exactly the
  same as the equalities within~$F_2$,
and any values occurring in both $F_1$ and $F_2$ occur
in the same positions in~$F_1$ and~$F_2$.
\end{itemize}
If $n$ and $n'$ are not far apart, we say
that they are \emph{near}.

Given a match $h$ of~$Q$ in the chase tree~$T$, its
\emph{augmented image} is the closure of its image
under least common ancestors, including by convention the root node. If $Q$ has size $k$ then
this has size $\leq 2k+1$.
For any two nodes $n$ and $n'$ in the augmented image,
we call $n$ the \emph{image parent of} $n'$ 
if $n$ is the lowest  ancestor
of~$n'$ in the augmented image.

\begin{lem} If $Q$ has a match $h$ in the final chase tree~$T$ of a tree-like
  chase sequence,
  then there is another tree-like chase sequence with final tree $T'$,
  and a match $h'$ of~$Q$ in~$T'$ with the property that if
$n$  is the image parent of~$n'$
then $n$ and $n'$ are near.
\end{lem}
\begin{proof}
We prove that given such an $h$ and $T$, we can
construct an $h'$ and $T'$ such that we decrease 
the sum of the depths of the violations.

If $n$ is far apart from~$n'$, then there are witnesses
$F_1$ and $F_2$ to this, corresponding to nodes $n_1$ and $n_2$ respectively.
   Informally, we will ``pull up'' the homomorphism
by replacing witnesses below $F_2$ with witnesses below $F_1$.
Formally, we create  $T'$ by first removing each
step of the chase proof  that generates a node that is a descendent
 of $n_1$.  Letting $T_1$ be the nodes in $T$ that do not lie below $n_1$, we 
will add nodes and the associated proof steps to $T'$.
Let $C_2$ be the chase steps in $T$ that generate a node below $n_2$, ordered 
as in~$T$, and let $T_2$ be the nodes produced by these steps.
We then add chase steps in $T'$ for each chase step in $C_2$. 
More precisely, 
we expand $T'$ by an induction on prefixes of $C_2$, building
$T'$ and a  partial function  $m$ from the domains
of facts in $\{n_2\} \cup T_2$ into the domain of facts  associated to
$n_1$ and its descendants in $T'$. The invariant is that $m$ preserves each fact of $T$ generated
by the chase steps in $C_2$ we have processed thus far in the induction, and that
$m$ is the identity on any values in $F_1$.
We initialize the induction by mapping the elements associated to $n_1$
to elements associated to $n_2$.
Our assumptions
on  $n_1$ and $n_2$ suffice to guarantee that we can perform such a mapping satisfying the invariant.
For the inductive case, suppose the next
chase step $s$ in $C_2$ uses linear TGD $\delta$, firing
on the fact associated to $v_i$ in $T$,  producing node $v_{i+1}$.
Then we perform a step $s'$ using
 $\delta$ and the fact associated to  $m(v_i)$ in $T'$.
If $\delta$ was a full TGD we do not modify $m$, while if it is a non-full 
TGD we  extend $m$ to map the generated elements of $s$ to the corresponding
elements of $s'$.
 We can thus form $h'$ by revising $h(x)$ when $h(x)$ lies
below  $n_1$, setting $h'(x)$ to $m(h(x))$.
Note that there could not have been any elements in the augmented
image of $h$ in $T$ that hang off the path between $n_1$
and $n_2$, since $n$ and $n'$ were assumed to be adjacent in the augmented image
  and the augmented image is closed under least common ancestors.

In moving from $T$ and $h$ to $T'$ and $h'$ we reduce the sum of the depths of
nodes in the image, while no new violations 
are created, since the image-parent relationships  are preserved.
\end{proof}

Call a match $h$ of~$Q$ in the chase \emph{tight} if it
has the property given in the lemma above. The \emph{depth}
of the match is the depth of the lowest node in its image.
The next observation, also due to Johnson and Klug, is that when the width
is bounded, tight matches cannot occur  far down in the tree:

\begin{lem} \label{lem:depthbound} If $\Sigma$ is a set of linear TGDs of width $w$
and the schema has arity bounded by $m$,
then any tight match  of size $k$  has all of its nodes at depth at most
  $k \cdot \card{\Sigma} \cdot (m+w)^w$.
\end{lem}

\begin{proof}
We claim that the length of the path  between a node $n$ of the image of the
  match and its image parent $n'$ 
 must be at most
  $\Delta \colonequals  \card{\Sigma} \cdot (m+w)^w$
  Indeed, every fact on the path was created by applying
a rule of~$\Sigma$: choosing such a rule~$\sigma$, the occurrences
  of variables in the head tell us which of the elements of facts created by the
  application of~$\sigma$ are necessarily equal. Specifically, the elements at
  positions corresponding to existential variables contain fresh values with
  equalities that are exactly as indicated; and the elements at exported
  positions contain at most $w$ distinct values, with equalities specified by
  the variable occurrences plus possibly additional equalities if some of the
  $w$ values are in fact equal. Thus, considering the values occurring in the
  fact of~$n'$ (at most~$m$), the status of a descendant fact can be
  characterized by:
  \begin{itemize}
    \item  the last rule used; this corresponds to a factor of~$|\Sigma|$
    \item For each of the exported elements (at most $w$), knowing which are equal
      to elements of~$n'$ or to some different element, i.e., each such exported
      element is either one of the $m$ elements of~$n'$ or some value in 
      $\{1 \ldots w\}$ used to represent the equality patterns between the elements that
      are not in~$n'$; this corresponds to a factor of $(m+w)^w$
  \end{itemize}

  Thus, after $\Delta$ steps, 
  there will be two elements which
repeat both the rule and the configuration of the values, which
would contradict tightness.
Since the augmented image contains the root, this implies the bound above.
\end{proof}

Johnson and Klug's result, 
generalized from IDs to linear TGDs,
follows from combining the previous two lemmas:
\begin{propC}[\cite{johnsonklug}]
  \label{prop:jkwidth}
For any fixed  $w \in \NN$, 
there is an $\np$ algorithm for query containment under linear TGDs of
width at most $w$.
\end{propC}
\begin{proof}
We know it suffices to determine whether there  is a match in
a chase proof, and the previous lemmas tell us that 
the portion of a chase proof required to find a match is not large.
We thus guess a tree-like chase proof where the tree consists
of $k$ branches of depth at most
$k \cdot \card{\Sigma} \cdot (m+w)^w$ for $k$ the query size,
along with  
a match in them, verifying the validity of the branches according to the
rules of $\Sigma$.
\end{proof}

We now give the  extension of this argument for bounded semi-width.
Recall from the body that a collection of
linear TGDs $\Sigma$ has \emph{semi-width} bounded
by $w$ if it can be partitioned as~$\Sigma = \Sigma_1 \cup \Sigma_2$
where $\Sigma_1$ has width bounded by $w$ and
the basic position graph of~$\Sigma_2$ is acyclic.
An easy modification of Proposition~\ref{prop:jkwidth}
now completes the proof of our semi-width result
(Proposition~\ref{prop:semiwidthclassic-general}):
\begin{proof}[Proof of Proposition~\ref{prop:semiwidthclassic-general}]
We revisit the argument of Lemma~\ref{lem:depthbound}, claiming a bound
  with an extra factor of $\card{\Sigma}$ in it.
As in that argument, it suffices
  to show that, considering the extended image of a tight match of~$Q$ in
  a chase proof, then the distance between any node~$n'$ of the extended image
  and its closest ancestor~$n$ is bounded, i.e., it 
  must be at most $|\Sigma|^2 \cdot (m+w)^w$.
Indeed, as soon as we apply a rule of~$\Sigma_1$ along
the path, at most $w$ values are exported, and so the
remaining path is bounded as before.
Since $\Sigma_2$ has an acyclic basic position graph, 
a value in~$n$ can propagate for at most $|\Sigma_2|$ steps 
when using rules of~$\Sigma_2$ only. Thus after
at most $|\Sigma_2|$  edges in a path
we will either have no values propagated (if we used only
rules from~$\Sigma_2$) or at most $w$ values (if we used
  a rule from~$\Sigma_1$). In particular, we cannot
have a gap of more than  $ | \Sigma_2|  \cdot 
  |\Sigma| \cdot (m+w)^w$
in a tight match, establishing our desired distance bound of $|\Sigma|^2 \cdot (m+w)^w$.
\end{proof}

\section{Supporting constants and multi-headed GTGDs}
\label{apx:constantmulti}

In this appendix, we make formal the claim from Section~\ref{sec:results} that
multi-headed GTGDs with constants in rule bodies can be encoded to single-headed GTGDs without
constants. Thus our $\exptime$ 
upper bound from Result~\ref{res:exptime} also
applies to multi-headed GTGDs which may feature constants (provided the constants are not
in rule heads).

We first formally define multi-headed GTGDs. Remember that a \emph{single-headed
TGD} was defined in Section~\ref{sec:prelims} as an FO sentence of the
following form:
$\forall \vec x ~ (\body(\vec x) \rightarrow \exists \vec y ~ A(\vec x, \vec
y))$.
A \emph{multi-headed TGD} is defined in the same way but as:
$\forall \vec x ~ (\body(\vec x) \rightarrow \exists \vec y ~ \head(\vec x, \vec
y))$
where $\head$ is a conjunction of atoms. As in the case of single-headed GTGDs,
we say that a multi-headed TGD is \emph{guarded} if
there is an atom in the body~$\beta$ which contains all variables occurring in~$\beta$.
Further, we define TGDs with \emph{constants} (single-headed or multi-headed) by
allowing atoms in TGDs to feature constants as well as variables. The
constants in question can also be used in the active domain of the instance $I_0$ given as
input to $\owqa$, and in the query $Q$ given as input to $\owqa$. However, \emph{we
disallow constants in the head of TGDs}: we discuss at the end of the appendix
why these are different.

The \emph{$\owqa$ problem with multi-headed TGDs with constants in rule bodies} is defined as follows:
given an instance $\inst_0$, a query $Q$ (possibly with constants), and a set of
guarded TGDs $\Sigma$ (which may be multi-headed, and may feature constants in
rule bodies), decide whether $\inst_0, \Sigma \models Q$ or not.

In this appendix, we show that Result~\ref{res:exptime} also holds in this
setting: for any constant number $a' \in
\mathbb{N}$, if the input GTGDs $\Sigma$ obey a side signature of
maximal arity~$a'$, then the $\owqa$ problem is in $\exptime$. We do this by
showing that we can rewrite the input $\Sigma$ to transform it to single-headed
GTGDs without constants while preserving the assumption that a bounded-arity side
signature is obeyed, so that we can then conclude by
Result~\ref{res:exptime}.

\paragraph*{Reducing to single-headed GTGDs.}
We first explain how to reduce from multi-headed to single-headed GTGDs:

\begin{lem}
  Let $\Sigma$ be a set of multi-headed GTGDs with constants over signature $\sign$ obeying a side signature
  $\sidesign$. We can rewrite $\Sigma$ in polynomial time to a set $\Sigma'$ of
  single-headed GTGDs with constants over a signature $\sign' \supseteq \sign$ such that
  $\Sigma'$ obeys side signature~$\sidesign$ and such that $\Sigma$ and $\Sigma'$
  are $\sign$-entailment-equivalent for $\owqa$.
\end{lem}

\begin{proof}
  We rewrite each GTGD of~$\Sigma$ separately.
  Let $\gdep$ be a multi-headed GTGD from~$\Sigma$, namely, $\gdep: \forall \vec x ~ (\body(\vec x) \rightarrow \exists \vec y ~ \head(\vec x, \vec
y))$.

We introduce a fresh predicate $P_\gdep$ in the signature, and replace~$\gdep$
  by several TGDs.
\begin{itemize}
\item The single-head GTGD $\gdep': \forall \vec x ~ (\body(\vec x) \rightarrow
\exists \vec y ~ P_\gdep(\vec x, \vec y))$
\item For each atom $A(\vec x, \vec y)$ in the head~$\head(\vec x, \vec y)$, the full linear TGD:
  $\gdep_A: \forall \vec x \vec y ~ P_\gdep(\vec x, \vec y) \rightarrow A(\vec
    x, \vec y)$.
\end{itemize}
We let $\Sigma'$ be the result of this transformation.
The transformation is clearly in polynomial time, and the side-signature
restriction is still obeyed because each new GTGD of~$\Sigma'$ either is linear or has the
same body as a GTGD of~$\Sigma$.
Further, it is clear that $\Sigma$ and $\Sigma'$ are
  $\sign$-entailment-equivalent.
\end{proof}

Notice that the transformation given in the proof above may increase the width
of GTGDs, because it creates GTGDs whose width is as large as the maximal number
of variables used in an atom of the head of a multi-headed GTGDs. While this is
not a problem to generalize Result~\ref{res:exptime}, it means that the same
transformation cannot be used to generalize Result~\ref{res:np}.

\paragraph*{Eliminating constants.}
We next explain how to reduce to GTGDs without constants.

\begin{lem}
\label{lem:constants}
  Let $\Sigma$ be single-headed GTGDs
  over signature $\sign$ which obey a side signature
  $\sidesign$
  and may feature constants in
  rule bodies. Let $I_0$ be an instance on~$\sign$, and let~$Q$ be a query on
  $\sign$ (possibly with constants).
  We can rewrite $\sign, \sidesign, \Sigma, I_0, Q$ in polynomial time to:
  \begin{itemize}
    \item new side signature relations $\sidesign'$ whose maximal arity is no
      greater than that of $\sidesign$;
\item the new side signature $\sidesign \cup \sidesign'$, and the new signature
  $\sign \cup \sidesign'$;
    \item a set $\Sigma_2$ of
  single-headed GTGDs without constants over the new signature $\sign \cup
      \sidesign'$ such that
  $\Sigma'$ obeys the new side signature~$\sidesign \cup \sidesign'$;
\item an instance $I_0'$ over the new signature,
\item a query $Q'$ without constants over the new signature.
  \end{itemize}
  Further, we have $I_0, \Sigma \models Q$ iff $I_0', \Sigma' \models Q'$.
\end{lem}

\begin{proof}
We use a standard technique for mimicking constants with unary predicates in guarded logics \cite{guardedconstants}.
  For each constant $c$ used in the GTGDs of~$\Sigma$ or in the query~$Q$, we
  introduce a fresh unary predicate $P_c$ which we add to the new side
  signature. We let the set $\sidesign'$ of new side signature predicates be
  the set of these unary predicates, which clearly satisfies the arity bound.

  We rewrite the instance $I_0$ to~$I_0'$ in the following way: for each
  constant $c$ that occurs in the active domain of~$\inst_0$, we add the new unary
  fact $P_c(c)$.

  We rewrite the query~$Q$ to~$Q'$ in the following way: for each constant $c$
  that occurs in~$Q$, we add a new variable $x_c$, replace $c$ by $x_c$, and add
  the atom $P_c(x_c)$.

  We rewrite the single-headed GTGDs $\Sigma$ in the following way: for each GTGD
$\forall \vec x ~ (\body(\vec x) \rightarrow \exists \vec y ~ A(\vec x, \vec
y))$, for each constant $c$ used in the $\body$, we introduce a new variable
$x_c$, replace $c$ by~$x_c$, and add a new atom $P_c(x_c)$. The result is still
a single-headed TGD; it is still guarded because the guard atom still contains all the variables (it
includes all pre-existing variables as well as all of the new variables); and it
now obeys the side-signature $\sidesign \cup \sidesign'$ because all atoms except the
principal atom of~$\body$ is either in $\sidesign$ or is an atom for
one of the relations $P_c$ which is in $\sidesign'$.

It is then clear that $I_0, \Sigma \models Q$ iff $I_0', \Sigma' \models Q'$.
\end{proof}

Notice that the transformation given in the proof above increases the number of
relations in the side signature. Again, while this is
not a problem to generalize Result~\ref{res:exptime}, it would be a problem to 
generalize Result~\ref{res:np}.

\paragraph*{Issues with constants in TGD heads.}
We last discuss why the translation in Lemma~\ref{lem:constants} cannot be
used as-is when GTGDs feature constants in rule heads. The problem is
that rule head with constants, e.g., $R(x, y) \rightarrow S(y, c)$, may
force us to create facts involving one specific element~$c$: this cannot be
replaced by an existentially quantified variable.

One alternative translation that can be used to allow constants in rule heads is
to enlarge the arity of each relation by~$N$, where $N$ is the number of
constants used; and store the domain elements that correspond to constants in
the $N$ extra positions. However, unlike the transformations in this
appendix, this would enlarge the arity of the side signature
relations, so it would not preserve the constant bound on the side signature
arity. We leave open the question of whether our $\exptime$ bound can be
extended to GTGDs with constants in the head of rules,
and also leave open the question of generalizing the $\np$ bound.

\end{document}